\documentclass{article}
\usepackage{amsmath,amssymb,amsfonts}
\usepackage[margin=1in]{geometry}
\usepackage{amsthm}
\usepackage{authblk}
\usepackage{algorithmic}
\usepackage{graphicx}
\usepackage{textcomp}
\usepackage{todonotes}
\usepackage{csquotes}
\usepackage{hyperref}
\usepackage{tikz}
\usetikzlibrary{calc,trees,positioning,arrows,chains,shapes.geometric,  decorations.pathreplacing,decorations.pathmorphing,shapes,matrix,shapes.symbols}
\usetikzlibrary{arrows.meta,positioning}
\usepackage[noabbrev,capitalise]{cleveref}
\usepackage[toc,page]{appendix}
\usepackage[font=scriptsize,margin=.3in]{caption}

\newcommand{\prob}[1]{\mathbb{P}\left(#1\right)}
\newcommand{\norm}[1]{\left\lVert#1\right\rVert}

\newcommand{\E}[1]{\mathbb{E}\left[#1\right]}
\newcommand{\abs}[1]{\left\lvert #1 \right\rvert}
\newcommand{\s}[2]{\langle #1,#2 \rangle}

\renewcommand{\hat}{\widehat}
\renewcommand{\vec}[1]{\boldsymbol{#1}} 
\newcommand{\R}{\mathbb{R}} 
\newcommand{\N}{\mathbb{N}} 
\newcommand{\ncp}{\kappa} 
\newcommand{\ind}{\mathbb{I}} 
\newcommand{\SR}{\varphi} 

\DeclareMathOperator*{\argmin}{arg\,min}

\newtheorem{Thm}{Theorem}
\newtheorem{Def}{Definition}

\providecommand{\keywords}[1]
{
  \small	
  \textbf{\textbf{Keywords: }} #1
}

\begin{document}
\title{Robust inference of cooperative behaviour of multiple ion channels in voltage-clamp recordings}
\author[1]{Robin Requadt}
\author[2]{Manuel Fink}
\author[2]{Patrick Kubica}
\author[2]{Claudia Steinem}
\author[1]{Axel Munk}
\author[1]{Housen Li\thanks{Email: housen.li@mathematik.uni-goettingen.de}}
\date{\today}
\affil{\small{Institute for Mathematical Stochastics, Georg August University of G\"ottingen, Germany}\protect\\
$^2$\small{Institute of Organic and Biomolecular Chemistry, Georg August University of G\"ottingen, Germany}}
\maketitle

\begin{abstract}
Recent experimental studies have shed light on the intriguing possibility that ion channels exhibit cooperative behaviour. However, a comprehensive understanding of such cooperativity remains elusive, primarily due to limitations in measuring separately the response of each channel. Rather, only the superimposed channel response can be observed, challenging existing data analysis methods. To address this gap, we propose IDC (Idealisation, Discretisation, and Cooperativity inference), a robust statistical data analysis methodology that requires only voltage-clamp current recordings of an ensemble of ion channels. The framework of IDC enables us to integrate recent advancements in idealisation techniques and coupled Markov models. Further, in the cooperativity inference phase of IDC, we introduce a minimum distance estimator and establish its statistical guarantee in the form of asymptotic consistency. We demonstrate the effectiveness and robustness of IDC through extensive simulation studies. As an application, we investigate gramicidin D channels. Our findings reveal that these channels act independently, even at varying applied voltages during voltage-clamp experiments. An implementation of IDC is available from GitLab.
\end{abstract}


\section{Introduction}
\label{sec:introduction}
{I}{on} channels and protein pores mediate the transport of ions and nutrients across biological membranes. These transmembrane proteins play a critical role in governing nearly every cellular process \cite{jordan2005fifty}. Notably, they are required for homeostasis and signal transduction, play a pivotal role in neurotransmission, and serve as important drug targets possessing significant potential for the pharmaceutical industry. 

\subsection{Analysis of voltage-clamp current traces}
The voltage-clamp technique \cite{neher1976single} is a long-established method to record ion transport processes through ion channels, which produces current traces with abrupt changes corresponding to channel openings and closings. 
The recovery of the underlying channel profile from the noisy and convoluted voltage-clamp measurement, known as \emph{idealisation}, is a topic of intensive research in the recent literature. The existing idealisation methods can be roughly divided into two groups. 

The first group of methods builds on the commonly believed Markovian dynamics of channel profiles (which remain valid apart for rare exceptions e.g.\ in \cite{fulinski1998non}), and assumes an underlying hidden Markov model \cite{ball1992stochastic}. The built-in parametric assumption of such models allows fast computation \cite{QAT23} as well as efficient statistical inference (via e.g.\ sampling methods \cite{de2001statistical,siekmann2011mcmc},  expectation-maximisation type algorithms \cite{qin2000hidden, ven00iii}, or heuristic approaches  \cite{heinemann1991open, schroeder2015resolve, yellen1984ionic}). However, hidden Markov models cannot incorporate various practically relevant noise models (e.g.\ violet, pink, or heterogeneous noise) as well as baseline fluctuations, both of which are common in ion channel recordings \cite{VenWKS98,ven98ii}. Although certain extensions exist \cite{diehn2019maximum}, this remains a challenge in general. 

The second group of methods is nonparametric in nature, often referred to as model-free methods, as they do not rely on any pre-assumed parametric models. Typical methods are based on thresholding techniques \cite{neher1976single, vandongen1996new} and also on the minimum description length \cite{gnanasambandam2017unsupervised}. Recently, a novel subgroup of automatic idealisation methods (based on multiscale testing on the residuals over a range of time scales and locations, and variational estimation in the form of complexity penalisation) has been established as being flexible in modelling ion channel recordings, and at the same time computationally efficient. Some examples include the jump segmentation multiresolution filter,  JSMURF~\cite{hotz2013idealizing}, its extensions to flickering events (i.e., events on fine time scales below the filter length), JULES~\cite{pein2018fully} and further adaptation to open channel noises, HILDE~\cite{HILDE}. Notably, MUSCLE~\cite{LiuLi24} makes little assumption on the noise component (assuming only the median of noise component centred at zero), and thus it is robust against different types of noises. Besides, MUSCLE maintains a high detection power in locating switches of states, thanks to its focus on local errors rather than global errors.

\subsection{Cluster of ion channels}\label{ss:intro_cluster}

These sophisticated data analysis methods focus almost exclusively on single-channel events. One possible reason might be that ion channels are traditionally viewed as entities acting autonomously. Recent experimental findings that leverage advancements in electrophysiology, structural biology, and imaging techniques have challenged this conventional understanding of ion channels. As revealed by several studies \cite{Molina2015,Moreno2016,clatot2017,Pfeiffer2020,Dixon2022}, certain classes of ion channels appear to operate in a coordinated, cooperative manner. The cooperative behaviour of ion channels is conjectured from the fact that ion channels of the same type can assemble into dense clusters, where they experience physical interactions. 

However, a solid understanding of the cooperative behaviour of multiple channels remains missing, mainly because we can only access the total conductivity of a channel ensemble, rather than the individual conductivity of each channel. This challenges existing idealisation methods, as they do not allow to infer from the total conductivity to that of the individual channels. In more mathematical terms, this problem can be viewed as a specific instance of inferring coupling relations from multiple time series with only access to the superposition of these time series. Despite extensive research on coupling relations in multivariate time series \cite{Dahl00} (where each time series is individually accessible), only two works \cite{CHUNG,vanegas2023} have addressed this issue by employing specific instances of parametric hidden Markov models. In a pioneering work~\cite{CHUNG}, Chung and Kennedy proposed a hidden Markov model (referred to as the Chung--Kennedy model), which is built on a convex combination of a fully coupled interaction and an independent interaction. Despite its great appeal to be simple, it is thus restrictive to model only positive cooperativity (meaning that a channel is likely to open/close given another one is open/closed) or independence (i.e., zero cooperativity).  Therefore, Vanegas et al.~\cite{vanegas2023} introduced a vector norm-dependent hidden Markov model based on the exchangeablity of individual channels and certain conditional independence (see also \cref{VND-model}), which allows for modelling a wide range of interactions including positive, negative and zero cooperativity. Unfortunately, both models do not allow the incorporation of many important noise models and the low-pass filter, which is essential in the voltage-clamp current recording processes.

\subsection{Contribution of this paper}

To overcome this lack of methodology, in this paper we propose IDC, a novel, robust statistical data analysis methodology, which infers the cooperative behaviour of multiple ion channels using only the electrophysiology (voltage-clamp) measurement of the total conductance of ion channels. IDC consists of three steps: (i) \underline{i}dealisation of voltage-clamp current measurements, (ii) \underline{d}iscretisation of idealised states into the numbers of open channels, and then (iii) \underline{c}ooperativity inference, see \cref{f:intro}. This three-step procedure IDC allows us to incorporate the recent advances in idealisation methods \cite{LiuLi24} and in coupled Markov models \cite{vanegas2023}, and meanwhile to infer the cooperative behaviour of multiple channels with statistical guarantee. Unlike the existing approaches (i.e.~\cite{CHUNG,vanegas2023}) that rely on parametric hidden Markov models, IDC not only exhibits robustness to modelling errors, but also handles practical challenges inherent in voltage-clamp measurements (such as baseline fluctuations, low-pass filtering, and various noise types). In essence, IDC can be loosely viewed as a \emph{semiparametric} hidden Markov model with the emission distributions being non-parametric. Besides, IDC has the potential to be applied to general scenarios other than ion channels, where the underlying mechanism can be modelled as a coupled Markov model. The main contribution of this paper can be summarised as follows. 

First, the proposed IDC takes into account the low-pass filter that is intrinsic to voltage-clamp techniques and it makes very little assumption on the noise component, thus allowing heterogeneous, non-Gaussian noises (e.g., pink or violet) and outliers. Besides, IDC is computationally efficient, e.g.\ taking a few minutes for 100k samples on a standard laptop ($\approx 3.5$ minutes for Intel Core i5 processor, 8 GB memory).

Second, we introduce a minimum distance estimator for the inference of the underlying model parameter, which determines the cooperative behaviour of ion channels. We further establish its statistical validity under general mild conditions (\cref{consistency}). This ensures the correctness of inferred cooperativity for sufficiently long ion channel recordings, assuming the (approximate) validity of our statistical model.

Third, our statistical data analysis provides strong evidence that multiple gramicidin channels gate independently, in accordance with the common belief about gramicidin gating. We further performed the same analysis for different applied voltages to validate our results. We found that the cooperativity of gramicidin gating is independent of the applied voltage, even though the dynamics of the channels change with the voltage.  

\begin{figure}[!ht]
\includegraphics[width=\textwidth]{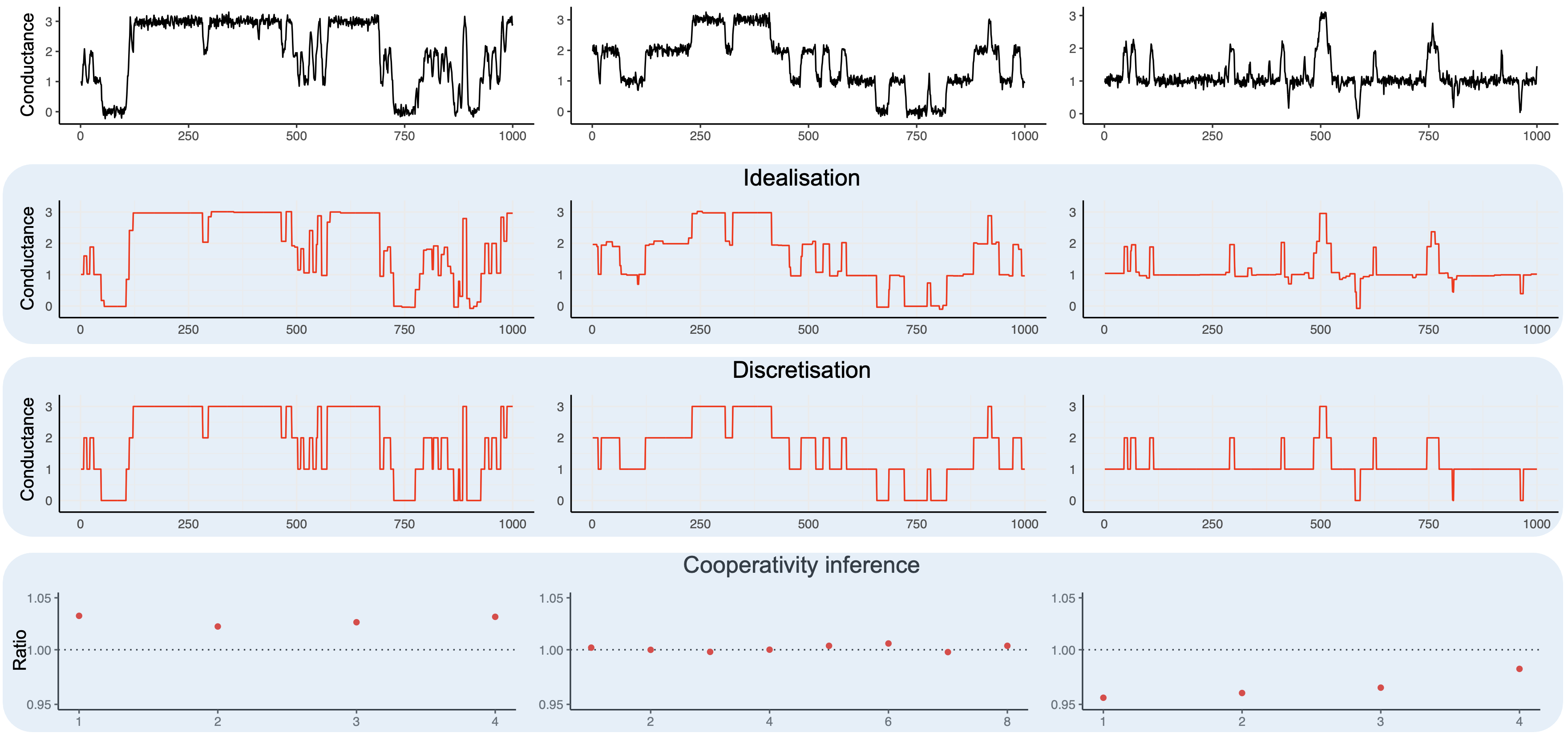}\\
\includegraphics[width=\textwidth]{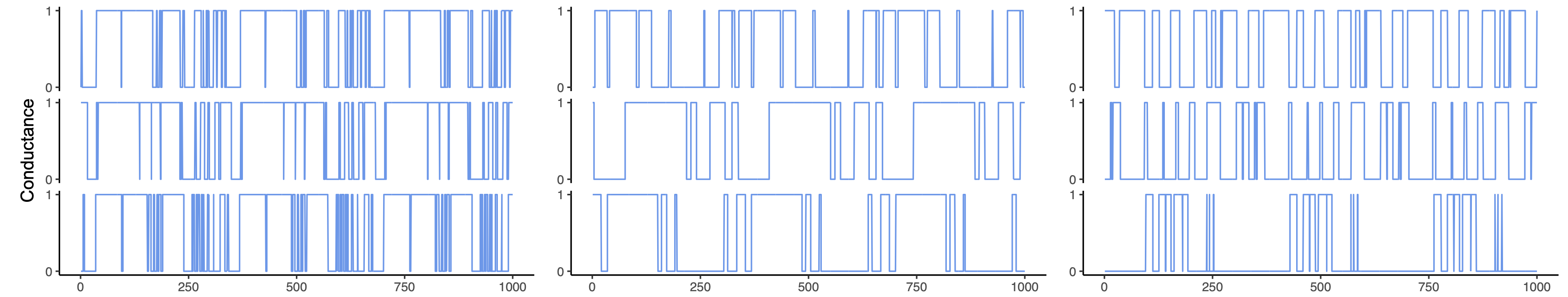}
\caption{Illustration of IDC (\underline{i}dealisation, \underline{d}iscretisation and \underline{c}ooperativity inference) on simulated data. From left to right are three scenarios of positive cooperativity, independence and negative cooperativity, respectively. The top panel plots the simulated ion channel recordings. The middle (second to fourth) panels present the results of IDC for each step. A positive (or negative) cooperativity is inferred if the estimated ratios in the fourth panel are all greater (or smaller) than one. An independent cooperativity is inferred if the estimated ratios are all equal to one. See \cref{behaviour}. Thus, IDC infers correctly the cooperative behaviour in all three scenarios. To shed light on the cooperative behaviour, we plot the \emph{simulated} profiles of individual channels in the bottom. Note that individual channels are inaccessible in practice and cannot be  directly recovered from the data.}
\label{f:intro}
\end{figure}

\subsubsection*{Organisation of the paper} We provide the experimental details in \cref{s:exp} and the statistical models in \cref{s:model}. The proposed IDC methodology is formally introduced in \cref{procedure} together with its statistical guarantee and practical considerations. We further validate IDC through simulation studies in \cref{simul} and then apply it to real data of gramicidin D channels in \cref{s:data}. We conclude the paper with discussions in \cref{s:discuss}.

\section{Experimental setup and measurements}\label{s:exp}
All voltage-clamp current measurements were recorded on black lipid membranes (BLMs) using the Orbit e16 (Nanion Technologies). 1,2-Diphytanoyl-\textit{sn}-glycero-3-phosphocholine (DPhPC) (Avanti Polar Lipids) was dissolved in octane (Sigma-Aldrich) to a concentration of 6 mg/mL. For the formation of BLMs in buffer (500 mM KCl, 10 mM HEPES, pH 7.4) the lipid-solvent solution was painted over the cavities of a MECA-16 chip (Ionera).  Membrane resistances $> 1\text{G}\Omega$ indicated a proper seal of the cavity. Gramicidin D (Sigma-Aldrich) was dissolved in ethanol and diluted to a concentration of 10 µM in water. Reconstitution of gramicidin channels into the BLMs was achieved by adding a small volume of the solution directly to the buffer solution above the membranes. Gramicidin current traces were recorded at fixed voltages. Holding potentials ranging from 100 mV to 50 mV and from $-100$ mV to $-50$ mV, both with increments of 10 mV were used for data collection. When multiple conducting gramicidin channels were observed, the resulting current traces were collected at a sampling rate of 19.53 kHz. All data were filtered by the built-in anti-aliasing filter of the e16 amplifier (Elements) with a kernel $\rho(\cdot)$ whose Fourier transform is $\mathrm{sinc}(\cdot/\SR)^{3}$ with $\SR$ the sampling rate.

\section{Statistical modelling}\label{s:model}

\subsection{Measurement model} \label{ss:meas_model}
In experiments with constant applied voltages (\cref{s:exp}), the conductance of ion channels stays at a fixed level unless an opening or closing event occurs. We thus model it by a piecewise constant function, $f:[0,t_{\max}] \to \R$ defined as
$$
f(t) = \sum_{j = 0}^{\ncp} c_j \ind( \tau_j \le t < \tau_{j+1}).
$$
Here $0 = \tau_0 < \tau_1 < \cdots <\tau_{\ncp} < \tau_{\ncp + 1} = t_{\max} < \infty$ denote the locations at which a gating event (i.e., switch of conductance levels) occurs, $\ncp \in \N$ the  number of gating events,  and $c_0, \ldots, c_{\ncp} \in \R$ the conductance levels. 

As earlier mentioned, the conductance profile of ion channels is measured by voltage-clamp current recordings. In this measurement process, there is an indispensable low-pass filter, which aims at eliminating high-frequency components to avoid aliasing and reducing background noise to improve the signal-to-noise ratio~\cite{filter}. Thus, the ion channel recordings $(y_1, \ldots, y_n)$ can be modelled as 
\begin{equation}\label{ideal}
    y_k=\left(\rho\ast f\right)(t_k)+\varepsilon_k,\quad k=1,\dots,n,
\end{equation}
where $\rho(\cdot)$ is the convolution kernel of the low-pass filter, $\varepsilon_i$ the noise, and $t_k = k/\SR$ the equidistant sampling location with $\SR$ the sampling rate. As a convention, we assume $t_n = t_{\max}$, namely, the full support of $f(\cdot)$ is sampled. The kernel $\rho$ is often known (\cref{s:exp}), and in otherwise case it can be well estimated from the data \cite{Tecu17,Chan22}. The noise $\varepsilon_k = (\rho \ast \mathcal{E})(t_k)$ is a convoluted version of an independent random process $\mathcal{E}(\cdot)$. We make no specific assumptions on the shapes of distributions of $\varepsilon_i$ except that for every $t$ the median of $\mathcal{E}(t)$ is zero. This allows heterogeneous noises (e.g., open channel noises), non-Gaussian and skewed noise components including high-frequent violet noises and long-tailed pink noises. We stress that accordingly the resulting methodology IDC is robust to outliers as well as to baseline fluctuations caused by accidental environmental influences. 

\subsection{Coupled Markov model} \label{VND-model}

We focus on a collection of multiple ion channels of the same type (e.g.\ gramicidin channels). For simplicity, we code the open state of a single channel as \enquote{$1$}, and the closed state as \enquote{$0$}. Let $L \in \N$ be the number of ion channels. Then the conductance profile of multiple ion channels forms a multivariate time series $(\vec{X}_k, k \in \N)$ with
$$
\vec{X}_k = ({X}_{k,1}, \ldots, X_{k,L})\; \in \;\{0,1\}^L
$$
where its $i$-th coordinate $(X_{k,i}, k \in \N)$ corresponds to the $i$-th channel. The individual conductance profile is experimentally inaccessible, and it is only possible to measure the total conductance profile of these $L$ channels, i.e.
\begin{equation}\label{e:sum}
S_k := \norm{\vec{X}_k}_1 = \sum_{i = 1}^L X_{k,i} \;\in\; \{0,1,\ldots,L\}.   
\end{equation}
Note that $\{S_k : k =1, \ldots, n\}$ is a discrete counterpart of $\{f(t_k) : k = 1, \ldots, n\}$, which is measured by voltage-clamp in model~\eqref{ideal},  in the sense that the conductance level $f(t_k)$ corresponds to the number $S_k$  of open channels at time $t_k$. 

The statistical problem is then to infer the cooperative behaviour of channels $(X_{k,i}: k \in \N)$, $i \in \{1,\ldots, L\}$, based on the information of $(S_k: k\in \N)$. We assume that 
$(\vec{X}_k : k \in \N)$ is a multivariate time-homogeneous Markov chain (see e.g.\ \cite{ball1992stochastic}). Further structural assumptions are necessary and possible, e.g.\ exploiting the exchangeability of individual channels. 

We consider particularly the vector norm-dependent Markov model (proposed in \cite{vanegas2023}) that aim at modelling cooperativity of ion channels. It is built on the structural assumptions in the form of {permutation invariance} (i.e., exchangeability) and conditional independence. In the context of multiple ion channels, permutation invariance means that individual channels are of identical importance, in particular, no existence of dominant channels. Conditional independence assumes that the ion channels behave independently of each other, given the states of all ion channels in the previous time step. 

\begin{Def}[Vector norm dependence \cite{vanegas2023}]\label{d:vnd}
A multivariate time-homogeneous Markov chain $(\vec{X}_k, {k\in \mathbb{N}_0 })$ with state space $\{0,1\}^L$ is a \emph{vector norm dependent Markov chain} (VND-MC), if, for any $\vec{x},\vec{z}\in\{0,1\}^L$, $k\in \N$, it satisfies: 
\begin{subequations}
\begin{enumerate}
    \item permutation invariance, i.e.
    \begin{equation}\label{e:perm_inv}
    \prob{\vec{X}_{k+1} = \vec{z} \mid \vec{X}_{k} = \vec{x}} 
    = \prob{\vec{X}_{k+1} = \vec{P}\vec{z} \mid \vec{X}_{k} = \vec{P}\vec{x}},
    \end{equation}
    for every permutation matrix $\vec{P}\in\{0,1\}^{L\times L}$,
    \item and conditional independence, i.e.
    \begin{equation}\label{e:cond_ind}
    \prob{\vec{X}_{k+1}=\vec{z}\mid \vec{X}_k=\vec{x}}=\prod_{i=1}^L \prob{X_{k+1,i}=z_i\mid {\vec{X}_k}={\vec{x}}}.
    \end{equation}
\end{enumerate} 
\end{subequations}
\end{Def}

The transition matrix of a VND-MC can be parameterized by $\vec{\theta}=(\lambda_0,\dots,\lambda_{L-1},\eta_1,\dots,\eta_L)\in [0,1]^{2L}$, with 
\begin{align*}
    \lambda_r & =\prob{X_{k+1, i}=0\mid X_{k,i}=0, \norm{\vec{X}_k}_1 = r}\\
    \text{and } \eta_{r +1} & =\prob{X_{k+1,i}=1\mid X_{k,i}=1, \norm{\vec{X}_k}_1 = r+1}.
\end{align*}
We stress that $\lambda_r,\eta_{r+1}$ are independent of $i \in \{1,\ldots, L\}$ and $k\in\N$, see \cite[Theorem~2.19]{vanegas2023}. As a consequence, the sum process $S_k = \sum_{i=1}^L X_{k,i}$ in~\eqref{e:sum} is also a time-homogeneous Markov chain, consistent with the observed Markovian dynamics in ion channel recordings (see \cref{ss:markov}). Recall that only $S_k$ is practically accessible. Thus, an important issue pertains to the possibility of determining the parameter vector $\vec{\theta}$ in VND-MC $(\vec{X}_k)_{k\in\N}$, from the transition probability matrix $\vec{Q}(\vec{\theta})$ of $S_k = \norm{\vec{X}_k}_1$, $k \in \N$. A prerequisite is that the mapping $\vec{\theta} \mapsto \vec{Q}(\vec{\theta})$ should be injective, in which case we say that VND-MC is \emph{identifiable.} In fact, the VND-MC is always identifiable for $L$ being an odd number and remains to be so for an even $L$ if we further assume $\lambda_{L/2} \ge 1-\eta_{L/2}$, or the opposite inequality, see \cite[Theorem~2.2]{vanegas2023}. 

The VND-MC is able to model a wide range of interactions of ion channels, including positive, negative, and zero cooperativity. They have a natural and direct interpretation of the biological behaviour of the ensemble of ion channels. 
\begin{Def}[Cooperative behaviour \cite{vanegas2023}]\label{behaviour}
   The VND-MC is \emph{positively cooperative}, if for all $r \in \{1,\dots,L-1\}$ 
\begin{subequations}
\begin{equation}\label{eq1}
    \frac{\lambda_0}{\lambda_r} >1\quad \text{and}\quad \frac{\eta_L}{\eta_r}>1 
\end{equation}
and \emph{negatively cooperative}, if for all $r \in \{1,\dots,L-1\}$
\begin{equation}\label{eq2}
   \frac{\lambda_0}{\lambda_r} < 1\quad \text{and}\quad \frac{\eta_{r+1}}{\eta_1} < 1.
\end{equation}
\end{subequations}
The VND-MC is \emph{zero cooperative} (or independent) if all such ratios in \eqref{eq1}--\eqref{eq2} are equal to 1. 
\end{Def} 

The individual coordinates of VND-MC are conditionally independent, see~\eqref{e:cond_ind}. The dependence between individual coordinates is through the state of the previous time, which introduces a temporal delay in their interactions. This is plausible for ion channels, where the interaction possibly occurs through physicochemical forces. In addition, as mentioned in \cref{ss:intro_cluster}, the range of cooperative behaviour among coordinates in the VND-MC is way wider than in the Chung--Kennedy model \cite{CHUNG}. Therefore, we decided to perform the inference of cooperative behaviour under the VND-MC model.  We stress, however, that our methodology IDC, in general, can be customised with other coupling models (including the Chung--Kennedy model) that incorporate different dependence structures.

\section{Methodology}\label{procedure}

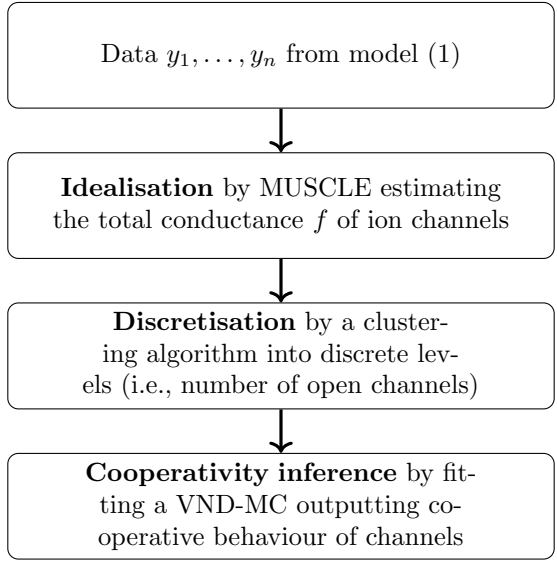
\begin{figure}
    \centering
    \tikzstyle{block} = [rectangle, draw, text width=20em,
    text centered, rounded corners, minimum height=4em]
\begin{tikzpicture}
 [node distance=2cm,
 start chain=going below,]
\node (n1) at (0,0) [block]  {Data $y_1,\dots,y_n$ from model~\eqref{ideal}};
\node (n2) [block, below of=n1] {\textbf{Idealisation} by MUSCLE estimating the total conductance $f$ of ion channels};
\node (n3) [block, below of=n2] {\textbf{Discretisation} by a clustering algorithm into discrete levels (i.e., number of open channels)};
\node (n4) [block, below of=n3] {\textbf{Cooperativity inference} by fitting a VND-MC outputting cooperative behaviour of channels};
\draw [->, very thick] (n1) -- (n2);
\draw [->, very thick] (n2) -- (n3);
\draw [->, very thick] (n3) -- (n4);
\end{tikzpicture}
\caption{Workflow of the proposed IDC procedure.}\label{IDC-procedure}
\end{figure}

We introduce a novel methodology (IDC) for the analysis of the interaction behaviour among multiple ion channels based only on electrophysiology measurements, see \cref{IDC-procedure}. It consists of three steps. First, we idealise the ion channel recordings by a recently proposed model-free method MUSCLE (\emph{\underline{m}ultiscale q\underline{u}antile \underline{s}egmentation \underline{c}ontrolling \underline{l}ocal
\underline{e}rror;} \cite{LiuLi24}). Second, the idealised states are grouped into the discrete levels $\{0,\dots,L\}$, where a level of $\ell$ corresponds to $\ell$ channels at open states. 
Third, based on idealised and discreted states, we infer the interaction behaviour of ion channels by fitting a VND-MC (\cref{d:vnd}). 

\subsection{Idealisation}

As a first step, we pre-process the ion channel recordings in~\eqref{ideal} by idealising the total conductance profile $f$, to separate the inference of cooperativity from the deconvolution problem. Recall that we make as little assumption as possible on the noise component in \cref{ss:meas_model} for the sake of permitting various types of noises, outliers, and baseline fluctuations. Most of the existing idealisation methods (e.g.\ JSMURF, JULES, and HILDE) rely on the assumption of Gaussian noise. Among the few exceptions is MUSCLE \cite[named D-MUSCLE therein]{LiuLi24}, which was developed under the same model assumption as model~\eqref{ideal}, and shows an attractive empirical performance. We thus employ it for this idealisation step. 

MUSCLE conducts statistical tests on the existence of events of state switching on a large collection of time windows at different locations and scales, and combines these multiple tests by controlling local errors defined on each individual constant segment of state sequences.  In the acceptance region of such multiple tests, MUSCLE selects the simplest candidate with the minimum number of state switchings as the final solution. Formally, MUSCLE is defined as a global optimal solution to a non-convex optimisation problem, i.e.,
\begin{align*}
\hat f  \;\in\; & \argmin_{g: [0,t_{n}] \to \R} \text{ number of state switches of } g\\
\text{subject to }\; & 
\underline{q}(\alpha, j-i)\le \sum_{k = i}^{j-1}\mathbf{1}_{\{y_k < (\rho\ast g)(t_k)\}} \le \overline{q}(\alpha, j-i)    \\
& \text{for all $(i,j)$ s.t.\ $g(\cdot)$ is constant on $[t_i - \varpi, t_j)$}
\end{align*}
where $\alpha \in (0,1)$ is the only tuning parameter, $\mathbf{1}_A$ the indicator function on set $A$, $\varpi$ the length of the support of kernel $\rho(\cdot)$, and $\underline{q}(\cdot,\cdot),\overline{q}(\cdot,\cdot)$ the explicitly known simple functions. 
This optimisation problem can be solved efficiently by a dynamic programming algorithm incorporating pruning steps together with a data structure of Wavelet Tree~\cite{Nav14}. 

It is well conceived that the control of local errors instead of global errors will increase the detection power in recovering events of state switches, and also the chance of including false positives. However, due to the minimisation of the number of state switches, MUSCLE enjoys the following finite sample guarantee on false positives
$$
\E{\frac{\max\{\hat\ncp - \ncp, 0\}}{\max\{\hat\ncp, 1\}}} \le \alpha,
$$
where $\hat\ncp$ is the number of state switches of $\hat f$. It provides practical guidance and statistical interpretation in selecting $\alpha$ for MUSCLE, which is a unique feature of this method. Asymptotical minimax optimality of MUSCLE in recovering $f$ is also established. See~\cite{{LiuLi24}} for further details. 

\subsection{Discretisation}\label{ss:discrete}
The total conductance levels of ion channels idealised by MUSCLE exhibit a group pattern (see e.g.\ \cref{fig3}), where a group represents approximately a certain number of open channels, i.e., $\{S_k\}$ in~\eqref{e:sum}. Thus motivated, we group the idealised conductance levels into a finite collection of discrete levels.  That is, we aim to divide $\mathcal{C} := \{\hat f (t_k), k = 1, \ldots, n \}$ into disjoint groups $\mathcal{C}_i$, $i = 0, \ldots, L$ such that $\bigsqcup_{i = 0}^L \mathcal{C}_i = \mathcal{C}$, with $\bigsqcup$ denoting the disjoint union. This can be naturally seen as the clustering problem, defined as follows 
\begin{align*}
&\min_{(\mathcal{C}_0,\dots,\mathcal{C}_L):{\bigsqcup}_{i=0}^L\mathcal{C}_i=\mathcal{C}}  \; \sum_{i = 0}^L \sum_{x \in \mathcal{C}_i} |x - \mu_i|^2 \\
\text{subject} &  \text{ to } \quad \mu_i - \mu_{i-1} \equiv \text{constant}, i = 1, \ldots, L
\end{align*}
where the constraint reflects the fact that individual ion channels are identical in their conductance characteristic. This clustering problem is one-dimensional and can be solved by a dynamic programming algorithm (cf.~\cite{WaSo11}). For every $i \in \{0, \ldots, L\}$, we code the group $\mathcal{C}_i$ simply by $i$ indicating the number of open channels, and refer it to as a \emph{discretised} conductance level. 

Recall the notation in \eqref{ideal}. It can be shown that $$
\mathbb{E}_f\left[\frac{1}{n}\sum_{k = 1}^n \bigl|f(t_k) - \hat f (t_k)\bigr|^2\right] \le C\frac{\ncp\log(n)}{n}, 
$$
for some constant $C$, under mild conditions (cf.\ \cite{LiGM19}). Since $f(\cdot)$ is assumed to follow a Markovian dynamics (i.e., a continuous time Markov chain), it holds that $\ncp = O_{\mathbb{P}}(n/\varphi)$, where the subscript $\mathbb{P}$ indicates the relation holds in probability. Then, $\mathbb{E}_f\bigl[\sum_{k = 1}^n \bigl|f(t_k) - \hat f (t_k)\bigr|^2/n\bigr]  = O_{\mathbb{P}}\bigl(\log(n)/\varphi\bigr)$, which converges to $0$ if $\log(n)/\varphi \to 0$, i.e., the sampling spacing vanishes at a rate faster than $1/\log(n)$ as $n\to \infty$. Thus, we may assume that MUSCLE together with the discretisation described above recovers well the sum process $(S_k, k = 1, \ldots, n)$ defined in \cref{VND-model}, and ignore the vanishing error in later inference on cooperative behaviour. 

\subsection{Fitting the VND-MC model}

\subsubsection{Minimum distance estimation}	
We will estimate the parameter vector $\vec{\theta}$ in the VND-MC based on the sum process $(S_k)_{k = 1}^n$ and then infer the cooperative behaviour by \cref{behaviour}. In~\cite{vanegas2023}, a customised Baum--Welch algorithm was proposed for vector norm dependent \emph{hidden} Markov models. However, there is no existing method for directly estimating the parameters of VND-MC. We thus introduce a minimum distance estimator. The intuition behind this is to find an estimator such that the corresponding estimated VND-MC is as \enquote{close} as possible to the underlying VND-MC. We measure the \enquote{closeness} in terms of a distance of the transition probabilities of the sum process $(S_k)_{k\in\N}$, and use the empirical transition frequencies of $(S_k)_{k = 1}^n$ as an estimator for the true transition probability matrix. Denote the transition probability matrix of $(S_k)_{k\in \mathbb{N}}$ by
$$
\vec{Q}(\vec{\theta}) = \bigl(q_{i,j}(\vec{\theta})\bigr)_{i,j=0}^L=\bigl(\mathbb{P}(S_{k+1}=j\mid S_k=i)\bigr)_{i,j=0}^L,
$$
and its empirical counterpart, i.e., the matrix of transition frequencies of $(S_k)_{k = 1}^n$, by $\hat{\vec{Q}}=(\hat{q}_{i,j})_{i,j=0}^L$. Thus, we define the \emph{minimum distance estimator} as 
\begin{equation}\label{e:mde}
\hat{\vec{\theta}}\;\;\in\;\; \argmin_{\vec{\theta}\in [0,1]^{2L}}\sum_{i=0}^L\sum_{j=0}^L\bigl(q_{i,j}(\vec{\theta})-\hat{q}_{i,j}\bigr)^2.
\end{equation}
The dependence of $q_{i,j}(\vec{\theta})$ on the parameter vector $\vec{\theta}=(\lambda_0,\dots,\lambda_{L-1},\eta_1,\dots,\eta_L)$ is known explicitly, given by 
$$ 
q_{i,j}(\vec{\theta})=
 \sum_{r=0}^{i} C_{i,j,r}\, \eta_i^{i-r}(1-\eta_i)^{r}\lambda_i^{L-j-r}(1-\lambda_i)^{j-i+r},
$$
where $C_{i,j,r} = \binom{i}{r}\binom{L-i}{j-i+r}$ if $j + r \ge i$ and $L\ge j+r $, and equals zero otherwise, see \cite[Proposition 2.21]{vanegas2023}.

In general, the optimisation problem~\eqref{e:mde} can be computationally difficult, as the objective function is not necessarily a convex function, and thus multiple local minima can exist. The initial values of the used optimisation algorithm are therefore of great importance for the algorithm to converge to a global minimum instead of a local minimum. If a proper initial value (discussed later in \cref{ss:practice}) is chosen, the optimisation problem can be tackled by standard optimisation methods, see e.g.\ \cite{NoceWrig06} for an overview. We employ an adaptive barrier method in the setting of linearly constrained optimisation (cf.\ \cite[Page~185ff]{optim}).

\subsubsection{Statistical guarantee}
A reasonable and mild requirement on $\hat{q}_{i,j}$ is that $\hat{q}_{i,j} \stackrel{\mathbb{P}}{\to}q_{i,j}(\vec{\theta}_*)$, as $n\to \infty$, where $\vec{\theta}_*$ is the true parameter vector of VND-MC, and $\stackrel{\mathbb{P}}{\to}$ denotes the convergence in probability. It is known (see e.g.\ \cite[page 56]{norris_1997}) that the empirical frequencies $\hat{q}_{i,j} \to q_{i,j}(\vec{\theta}_*)$ almost surely if the Markov chain $(S_k)_{k \in \mathbb{N}}$ is irreducible, which is the case in our ion channel setup. From this, we deduce a consistency guarantee for the proposed minimum distance estimator.

\begin{Thm}[Consistency]\label{consistency}
Assume that as $n\to \infty$,
\[
\hat{q}_{i,j}\stackrel{\mathbb{P} }{\to}q_{i,j}(\vec{\theta}_*)\quad \text{for } i,j = 0,\ldots, L.
\]
If the VND-MC is identifiable (\cref{VND-model}), then 
\[
\hat{\vec{\theta}}\stackrel{\mathbb{P}}{\to} \vec{\theta}_*,\quad\text{as }n\to \infty.
\]
\end{Thm}

\begin{proof}Let $\norm{\cdot}_{\mathrm{F}}$ denote the Frobenius norm and $\s{\cdot}{\cdot}_{\mathrm{F}}$ the corresponding Frobenius inner product. Then, we have
  \[
  \norm{\vec{Q}(\vec{\theta})-\hat{\vec{Q}}}^2_{\mathrm{F}}= \sum_{i=0}^L\sum_{j=0}^L\bigl(q_{i,j}(\vec{\theta})-\hat{q}_{i,j}\bigr)^2.
  \]
      Note that $[0,1]^{2L}$ is compact and the entries of $\vec{Q}(\vec{\theta})$ are continuous functions of $\vec{\theta}$. This combined with the injectivity (due to the identifiability assumption) of the mapping $\vec{\theta}\mapsto \vec{Q}(\vec{\theta})$ implies that for any $\varepsilon>0$ 
      \[
      \inf_{\vec{\theta}\in [0,1]^{2L}:\norm{\vec{\theta}-\vec{\theta}_*}\ge \varepsilon}\norm{\vec{Q}(\vec{\theta})-\vec{Q}(\vec{\theta}_*)}_{\mathrm{F}}> 0.
      \] 
      We will show 
      \begin{equation}
      \sup_{\vec{\theta}\in [0,1]^{2L}}\abs{\norm{\vec{Q}(\vec{\theta})-\hat{\vec{Q}}}^2_{\mathrm{F}}-\norm{\vec{Q}(\vec{\theta})-\vec{Q}(\vec{\theta}_*)}^2_{\mathrm{F}}} \stackrel{\mathbb{P} }{\to}0. \label{uniform_conv}
      \end{equation} 
      To this end, we have by the Cauchy--Schwarz inequality
      \begin{align*}
       &\abs{\norm{\vec{Q}(\vec{\theta})-\hat{\vec{Q}}}^2_{\mathrm{F}}-\norm{\vec{Q}(\vec{\theta})-\vec{Q}(\vec{\theta}_*)}^2_{\mathrm{F}}}\\
       =&\abs{\norm{\hat{\vec{Q}} }_{\mathrm{F}}^2-\norm{\vec{Q}(\vec{\theta}_*) }_{\mathrm{F}}^2+2\s{\vec{Q}(\vec{\theta})}{\vec{Q}(\vec{\theta}_*)-\hat{\vec{Q}} }_{\mathrm{F}}}\\
    \le& \abs{\norm{\hat{\vec{Q}} }_{\mathrm{F}}^2-\norm{\vec{Q}(\vec{\theta}_*) }_{\mathrm{F}}^2}+2(L+1) \norm{\hat{\vec{Q}}-\vec{Q}(\vec{\theta}_*)}_{\mathrm{F}},
      \end{align*}
      where we used $\sup_{\vec{\theta}\in [0,1]^{2L}}\norm{\vec{Q}(\vec{\theta})}_{\mathrm{F}}\le L + 1$. The  continuous mapping theorem and the fact that for all $i,j=0,\dots,L,$ 
      \[
      \hat{q}_{i,j}\stackrel{\mathbb{P}}{\to}q_{i,j}(\vec{\theta}_*)
      \]
      imply that
      \[
      \norm{\hat{\vec{Q}} }_{\mathrm{F}}^2-\norm{\vec{Q}(\vec{\theta}_*) }_{\mathrm{F}}^2 \stackrel{\mathbb{P}}{\to}0\; \text{ and }\; \norm{\hat{\vec{Q}}-\vec{Q}(\vec{\theta}_*)}_{\mathrm{F}}\stackrel{\mathbb{P}}{\to} 0.
      \]
It yields (\ref{uniform_conv}) and by \cite[Theorem 5.7]{v} finishes the proof. 
\end{proof}

We stress that, in general, minimum distance estimators are favoured due to their robustness compared to maximum likelihood estimators, see \cite{MDE1,MDE2}.

\subsection{Practical issue}\label{ss:practice}

\subsubsection{Choice of $L$} \label{choose_L}
In practice, the number $L$ of ion channels is unknown and needs to be determined from the data. This is a difficult model selection problem since various information criteria, as well as cross-validation procedures, do not perform well, as empirically demonstrated in \cite{vanegas2023}. Following the suggestion in \cite{vanegas2023}, we select $L$ by simply counting the distinct groups in idealised conductance levels. Of course, the selected $L$ may likely be smaller than the true number of channels, as not every channel would open at the same time. This is even more prevalent in case the ion channels behave negatively cooperatively, see \cref{simul}. However, the qualitative interaction behaviour (\cref{behaviour}) of a VND-MC remains stable over a range of choices of $L$, see  \cref{number_channels,data_analysis} and also \cite{vanegas2023}. 
 
\subsubsection{Choice of initial values} \label{initialization} 
The number $L$ of channels is usually not large, meaning the parameter space $[0,1]^{2L}$ is not very high dimensional. Thus, in the computation of the minimum distance estimator, we choose the initial values based on an exhaustive grid search approach. More precisely, we consider a finite set $G\subseteq [0,1]^{2L}$ and choose the initial vector $\vec{\theta}_{\text{init}}$ as the one in $G$ that gives the smallest objective value in \eqref{e:mde}. A finer grid $G$ increases the chance to find the initial values that lead to a global optimum of \eqref{e:mde} but comes at the cost of a larger computation time. In all our data analysis (including simulation) we used $G=\{0.1,0.2,\dots,0.9\}^{2L}$ since further refined grids did not lead to improvements to solutions of the optimisation problem \eqref{e:mde}. 

\subsection{Software}\label{ss:practice}
The implementation of IDC in {\it R} language, together with the {\it R} codes for reproducing the simulation study and real data analysis in this paper, is available from GitLab (\url{https://gitlab.gwdg.de/requadt/idc}). The three steps of IDC are implemented separately in a modular fashion. Hence, it allows for instance to employ a different idealisation method or a different coupling model in cooperativity inference.  

The first idealisation step is performed via the {\it R} package \texttt{MUSCLE}, available from GitHub (\url{https://github.com/liuzhi1993/muscle}). The discretisation in the second step is by a customised dynamic programming algorithm that adapts to the constraint on conductance levels (cf.\ \cref{ss:discrete}). In the third step of cooperativity inference, we compute the minimum distance estimator for the VND parameters using the {\it R} function \texttt{constrOptim}, available in the pre-installed {\it R} package \texttt{stats}, together with an exhaustive search of initial values (cf.\ \cref{initialization}). Finally, we calculate the ratios stated in \cref{behaviour} and conclude the cooperativity inference.

\section{Simulation study}\label{simul}
In this section, we examine the empirical performance of the proposed IDC methodology in various simulation setups that pertain to ion channel recordings.

\subsection{Comparison with VND-HMM} \label{vnd_hmm}
We consider three different scenarios, corresponding to zero, positively and negatively cooperative behaviours of ion channels. To benchmark the performance, we compare the proposed IDC to the hidden Markov model approach introduced in \cite{vanegas2023}, denoted by HMM, with slight abuse of abbreviation. 

For a fixed parameter vector $\vec{\theta}$, we simulate a VND-MC of length 1,200, with $L = 2$ channels.  Based on it, we create three scenarios by adding different types of noise: (i) Gaussian white noise with standard deviation $0.1$, (ii) Cauchy noise with location parameter $0$ and scale parameter $0.05$, and (iii) the mixture of these two with weights $0.85$ (Gaussian component) and $0.15$ (Cauchy component). The Cauchy distribution is employed to emulate outliers, which are commonly observed in real data applications. Finally, we simulate the data sets according to model~\eqref{ideal} by incorporating a low-pass filter of Bessel type. 

\Cref{simulation_1,simulation_2,simulation_3} show the comparison results in zero, positively and negatively cooperative scenarios, respectively. The proposed IDC outperforms the HMM across all three scenarios when the added noise is Cauchy distributed or follows the stated mixture distribution. In the case of Gaussian noise, both approaches have a similar performance. We can also see that IDC performs nearly the same for the three scenarios, indicating its robustness to noise distributions.  

\begin{figure}[!t]
\centerline{\includegraphics[width=.6\textwidth]{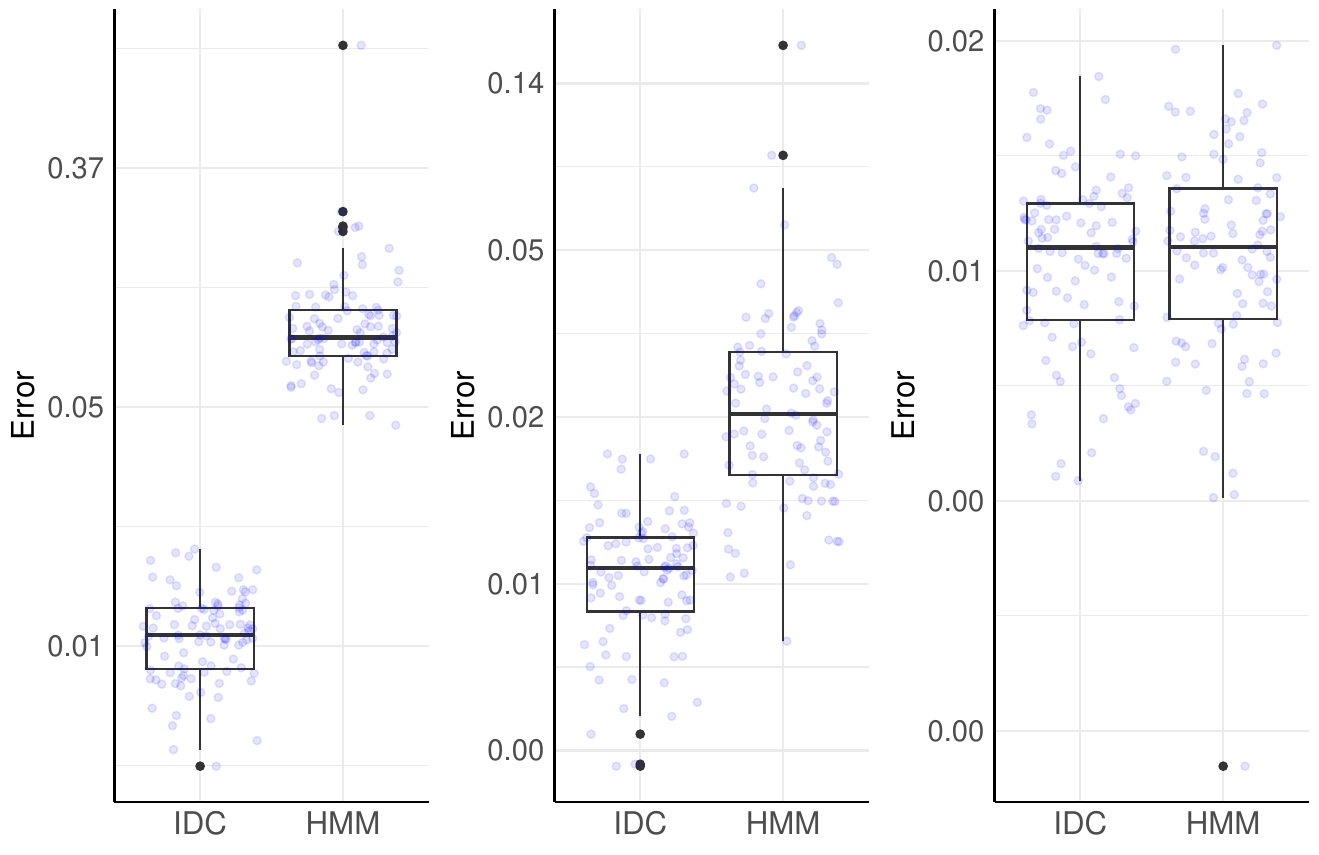}}
\caption{Comparison of the proposed IDC and the HMM \cite{vanegas2023} on two zero cooperative channels. Three panels differ only in noise distributions: left is $\mathrm{Cauchy}(0.05)$, middle the mixture distribution $0.85\mathcal{N}(0,0.1^2)+0.15\mathrm{Cauchy}(0.05)$, and right  $\mathcal{N}(0,0.1^2)$. In each panel, the $\ell_2$ errors in the estimation of the true parameter vector $(\lambda_0,\lambda_1,\eta_1,\eta_2)=(0.99,0.99,0.99,0.99)$, over $100$ repetitions, are summarised as box plots, with the error of each repetition shown by jittered purple dots. The vertical axis is on the log scale. }
\label{simulation_1}
\end{figure}

\begin{figure}[!t]
\centerline{\includegraphics[width=.6\textwidth]{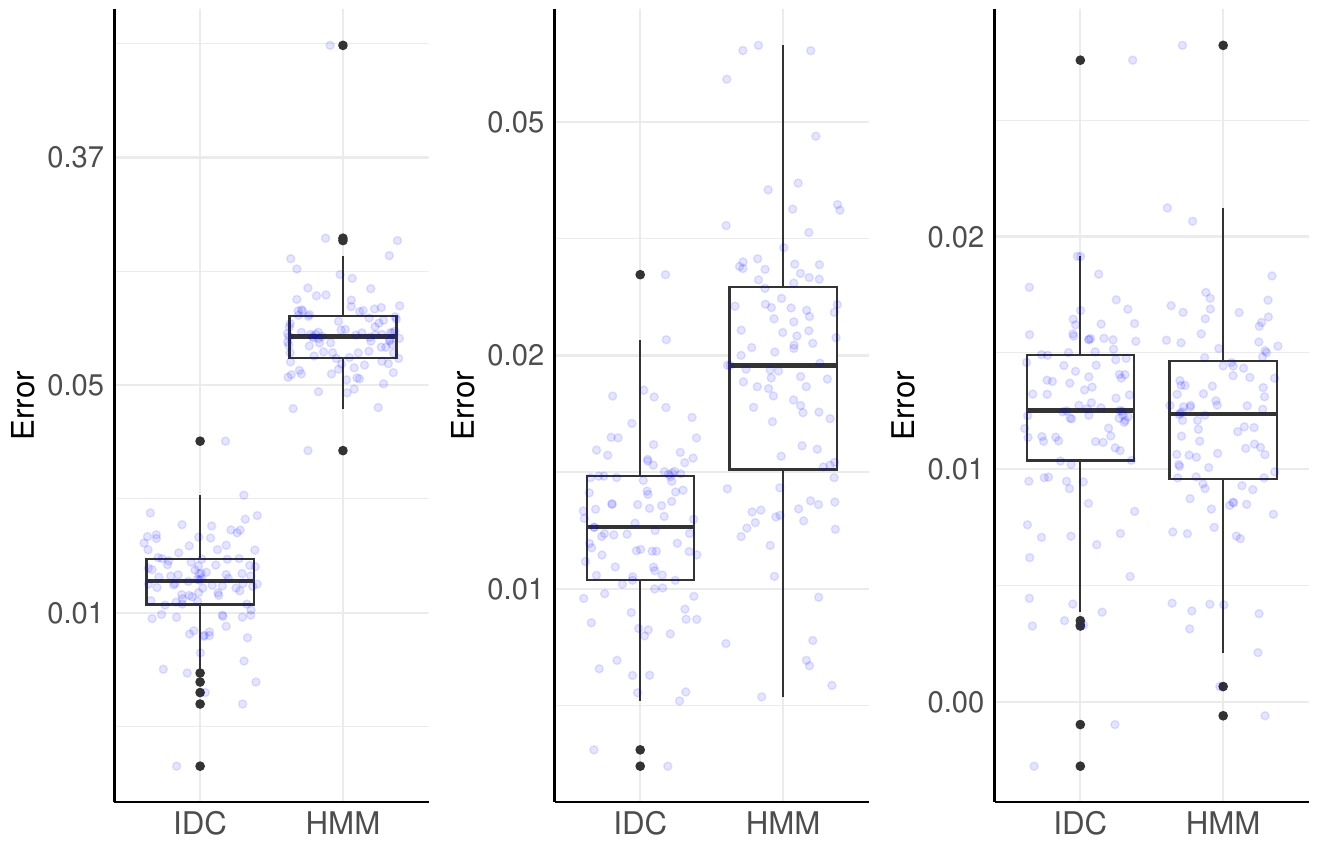}}
\caption{Comparison of the proposed IDC and the HMM \cite{vanegas2023} on two positively cooperative channels. Three panels differ only in noise distributions: left is $\mathrm{Cauchy}(0.05)$, middle the mixture distribution $0.85\mathcal{N}(0,0.1^2)+0.15\mathrm{Cauchy}(0.05)$, and right  $\mathcal{N}(0,0.1^2)$. In each panel, the $\ell_2$ errors in the estimation of the true parameter vector $(\lambda_0,\lambda_1,\eta_1,\eta_2)=(0.99,0.985,0.985,0.99)$, over $100$ repetitions, are summarised as box plots, with the error of each repetition shown by jittered purple dots. The vertical axis is on the log scale. }
\label{simulation_2}
\end{figure}

\begin{figure}[!t]
\centerline{\includegraphics[width=.6\textwidth]{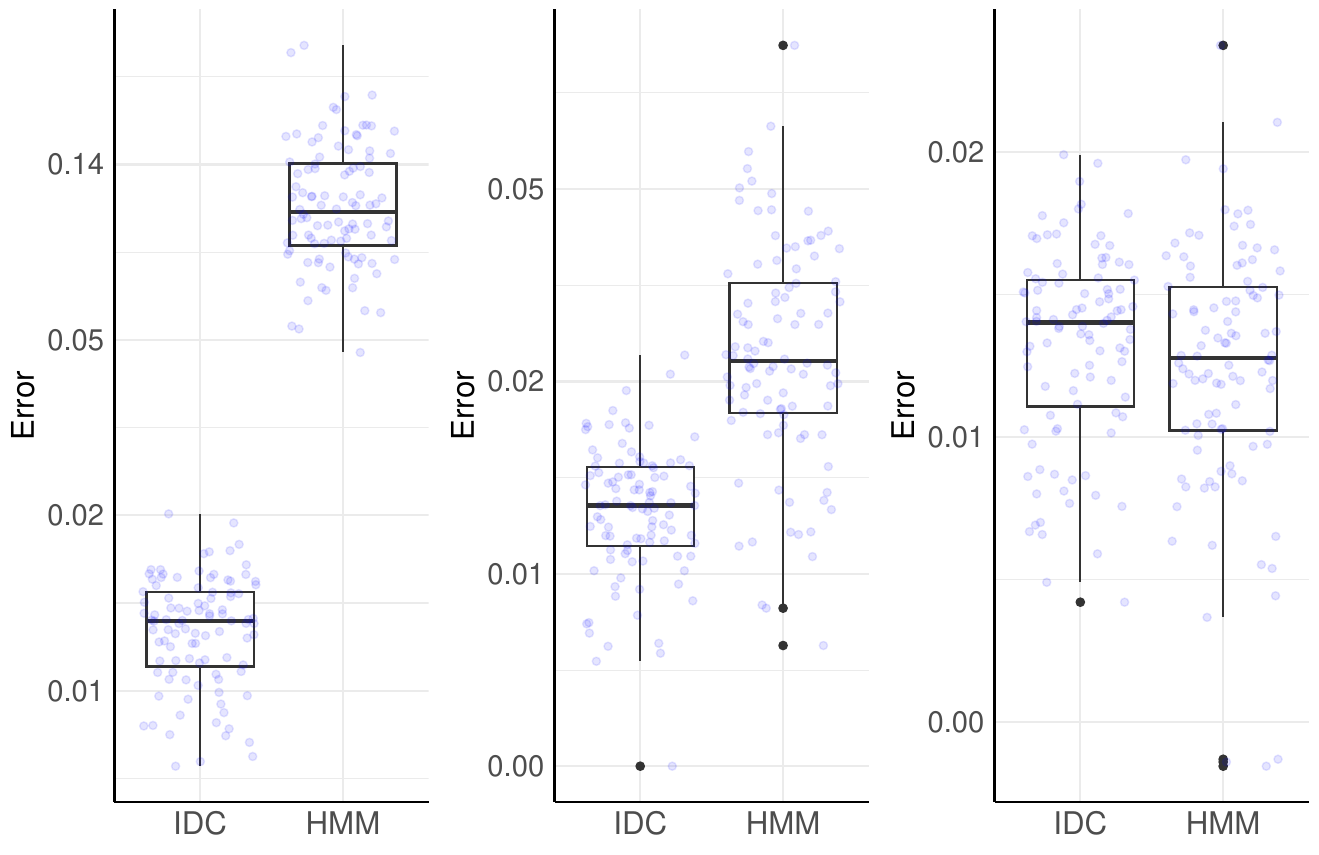}}
\caption{Comparison of the proposed IDC and the HMM \cite{vanegas2023} on two negatively cooperative channels. Three panels differ only in noise distributions: left is $\mathrm{Cauchy}(0.05)$, middle the mixture distribution $0.85\mathcal{N}(0,0.1^2)+0.15\mathrm{Cauchy}(0.05)$, and right  $\mathcal{N}(0,0.1^2)$. In each panel, the $\ell_2$ errors in the estimation of the true parameter vector $(\lambda_0,\lambda_1,\eta_1,\eta_2)=(0.985,0.99,0.99,0.985)$, over $100$ repetitions, are summarised as box plots, with the error of each repetition shown by jittered purple dots. The vertical axis is on the log scale. }
\label{simulation_3}
\end{figure}

\subsection{Influence of $L$} \label{number_channels}
The goal is to investigate the performance of IDC when the number of channels is wrongly specified. We set $L = 20$, and simulate three VND-MCs of length 100,000 with the parameter vector $\vec{\theta}\in\R^{2L} \equiv \R^{40}$ taking values of 
\begin{align*}
 &(0.99,\dots,0.99) && \text{for zero,}\\   
 &(0.99,0.98,0.98,\dots,0.98,0.98,0.99)&&\text{for positively,}\\
 &(\underbrace{0.98,0.99,\dots,0.99}_{\text{20 entries}},\underbrace{0.99,0.98\dots,0.98}_{\text{20 entries}}) && \text{for negatively}
\end{align*}
cooperative behaviours, respectively. 

The estimated values of $L$ are shown in \cref{simulation_L}. In the negatively cooperative scenario, the estimated number of channels is on average smaller than in the zero and positively cooperative scenarios. In the latter two scenarios, the estimated number of channels is relatively close to the true number of channels, though there is often an underestimation by no more than 3 channels. 

\begin{figure}[!t]
\centerline{\includegraphics[width=.6\textwidth]{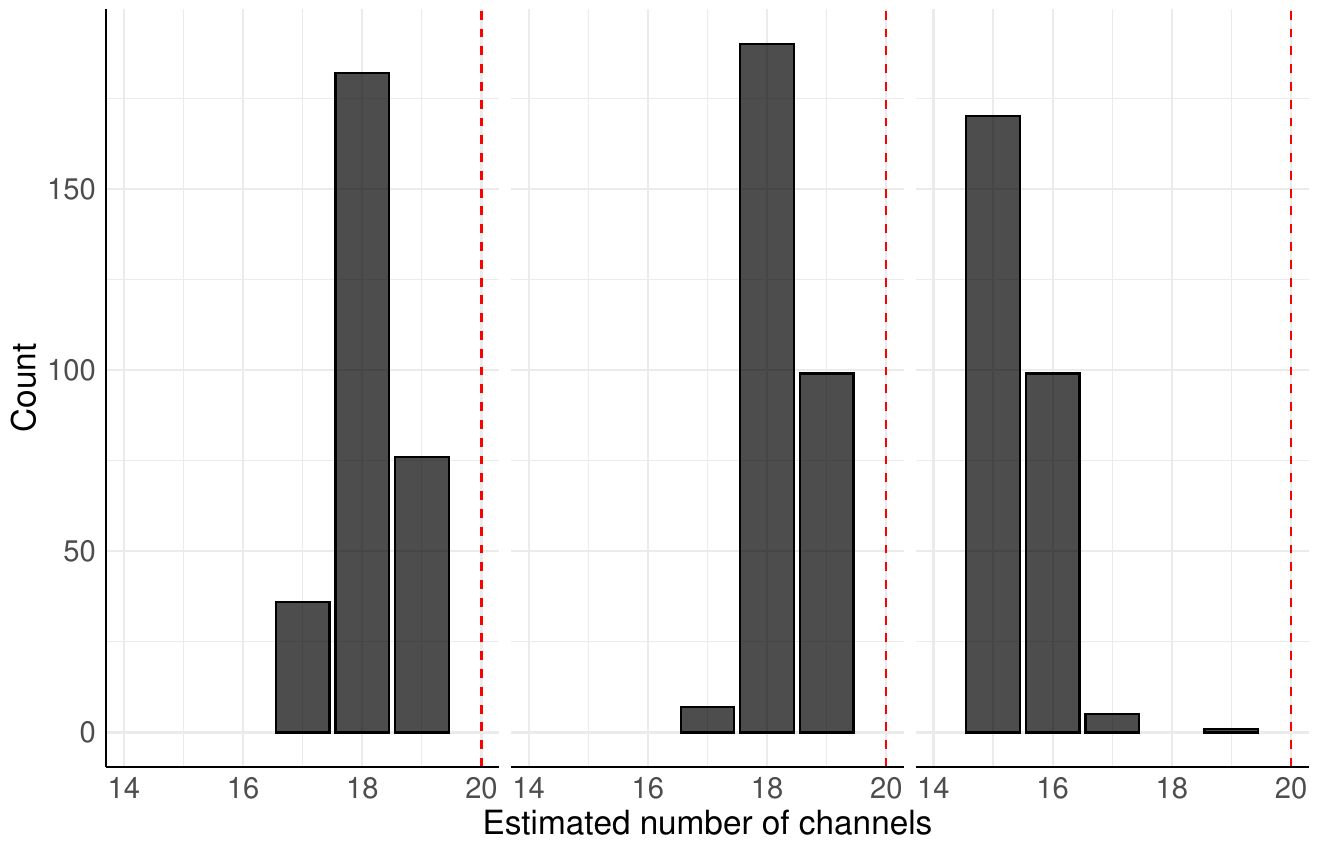}}
\caption{Histogram of the estimated number of channels over $300$ repetitions. From left to right panels are zero, positively and negatively cooperative scenarios, respectively. The true value of $L=20$ is marked by vertical dashed lines.}
\label{simulation_L}
\end{figure}

To see the influence of underestimated $L$ on the inference of cooperative behaviour, we show the estimated values of the ratios, introduced in \cref{behaviour}, in \cref{simulation_ratios}. In the scenario with zero cooperativity, the estimated ratios are close to one on average, and in the scenario with positive cooperativity, the estimated ratios appear to concentrate at a value greater than one (as indicated by the sharp peak of the histogram in the middle panel). This shows that our IDC approach is robust to the underestimation of the number of ion channels in both zero and positively-cooperative scenarios. However, in the scenario of negative cooperativity, there is a tendency for the estimated ratios to become greater than one, which should have been smaller than one by \cref{behaviour}. It suggests that IDC might miss the detection of negatively cooperative behaviours, mainly due to the severe underestimation of $L$. This is also in line with the observations in \cite{vanegas2023}.

\begin{figure}[!t]
\centerline{\includegraphics[width=.6\textwidth]{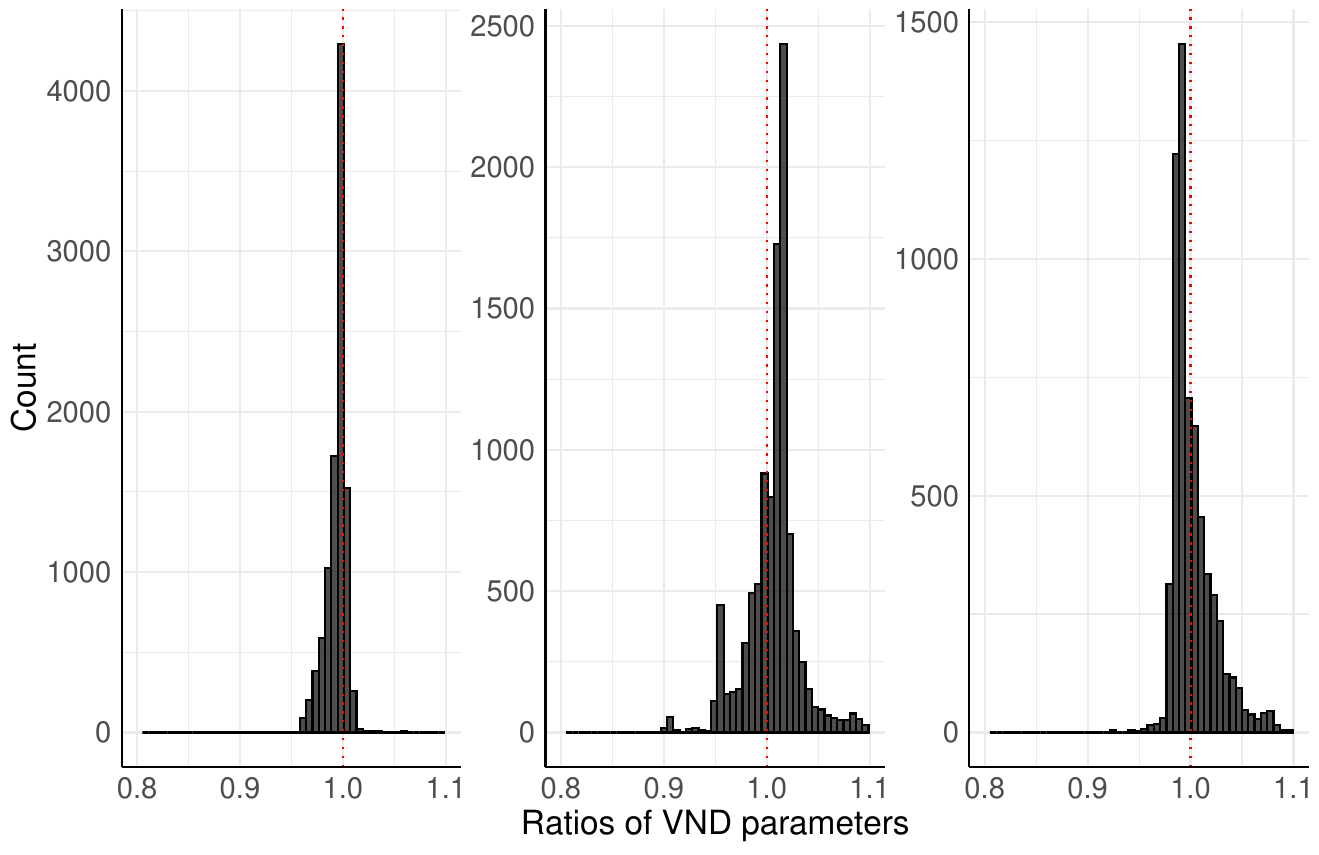}}
\caption{Histogram of the estimated values of the ratios in \cref{behaviour}, over 300 repetitions. The true ratios are equal to one in the zero cooperative scenario (left), greater than one in the positively cooperative scenario (middle) and smaller than one in the negatively cooperative scenario (right).}
\label{simulation_ratios}
\end{figure}

\section{Real data analysis}\label{s:data}

Gramicidin D is a short $\beta$-helical peptide known to form conducting channels in lipid membranes as a result of transbilayer dimerisation \cite{kondrashov2018membrane}. It can be readily reconstituted in lipid membranes and its electric response is fast and easy to detect. These characteristics make gramicidin an ideal test bed for the proposed IDC methodology (\cref{procedure}). We thus apply IDC to the gramicidin channel recordings (\cref{s:exp}) and deduce the cooperative behaviour of gramicidin ion channels. Additionally, we explore the influence of the applied voltage on the channel behaviour.

\subsection{Markovian assumption}\label{ss:markov}

As a validation of our modelling assumption, we first examine whether the joint conductance profile of multiple ion channels satisfies the Markov property (i.e., Markovian dynamics). Recall that a time-homogeneous time series $(S_k)_{k\in \mathbb{N}}$ taking values in $\{0, \ldots, L\}$ satisfies the \emph{Markov property}, if for any integer $j\ge 2$ and $s_1,\dots,s_j\in \{0,\dots,L\}$ it holds
\begin{equation}\label{markov}
    \prob{S_j=s_j\mid S_{j-1}=s_{j-1},\dots,S_1=s_1}
    =\prob{S_j=s_j\mid S_{j-1}=s_{j-1}}.
\end{equation}
That is, the time series in the current step only depends on the state in the previous time step rather than the whole past.

We utilise a commonly used chi-square-based statistical test~\cite{teste_mk} which inspects whether empirical transition matrices are statistically compatible with the relation \eqref{markov}. The implementation is provided in the {\it R} package \texttt{markovchain} (via function \texttt{verifyMarkovProperty}) on CRAN. For a gramicidin data set of length 99,999, corresponding to a recording of 5 seconds, the resulting $p$-value is $1.0$. It suggests that the gramicidin data indeed satisfies the Markov property with high statistical confidence. We further investigate the dwell times of our gramicidin data in~\cref{fig_sim3} (left panel), which appear to follow an exponential law. This is also in agreement with the Markov property. 

In addition, we fit a VND-MC to the gramicidin data and simulate the trace of the same length from the fitted VND-MC. The dwell times of the simulated trace are in~\cref{fig_sim3} (right panel). As shown, the distributions of dwell times coincide well with the gramicidin data and the simulated trace, indicating that VND-MC fits nicely with our gramicidin data.

\begin{figure}[!t]
\centerline{\includegraphics[width=.6\textwidth]{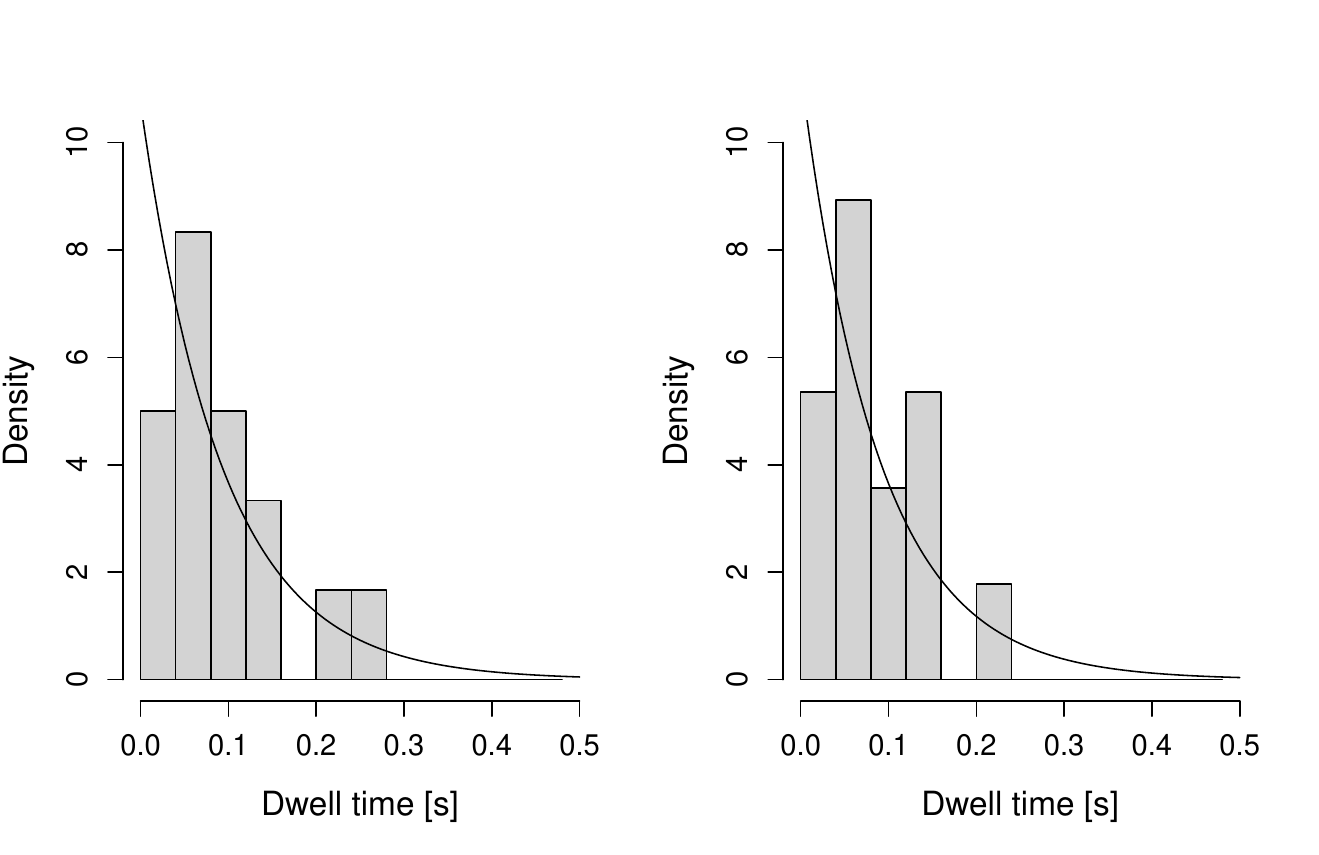}}
\caption{Histogram of dwell times of state 1 for the gramicidin data (left) and the simulated data (right). The curves show the exponential fits with the rate parameter equal to the inverse of the average of dwell times. }
\label{fig_sim3}
\end{figure}

\subsection{Independent behaviour of gramicidin channels}\label{data_analysis}

As the first step of IDC, the idealisation obtained by MUSCLE is presented in \cref{fig1}. In the second step, we determine the number of channels by plotting a histogram of idealised levels of conductance in \cref{fig3}, which exhibits four groups. Thus, we select $L = 3$ and group the idealised conductance levels into discretised levels $\{0,1,2,3\}$. In the third step, we estimate the parameter vector $\vec{\theta}$ of VND-MC by our minimum distance estimator, and obtain
\begin{equation*}
\hat{\vec{\theta}}=\bigl(0.9999560, 0.9999561 ,0.9998173,
0.9998505, 0.9998387, 0.9998944\bigr).
\end{equation*}
Based on this, the estimates for the ratios in \eqref{eq1} and \eqref{eq2} are equal to one up to an additive term with an absolute value smaller than $1.5\cdot 10^{-4}$. Note that by~\cref{consistency}  the minimum distance estimator should be close to the true parameter. Thus, the result of our analysis through IDC implies that multiple gramicidin channels behave \emph{independently} (i.e., zero cooperatively) of each other. This is in agreement with the common belief about gramicidin gating, see e.g.\ \cite{kondrashov2021peptide}.

\begin{figure}[!t]
\centerline{\includegraphics[width=.6\textwidth]{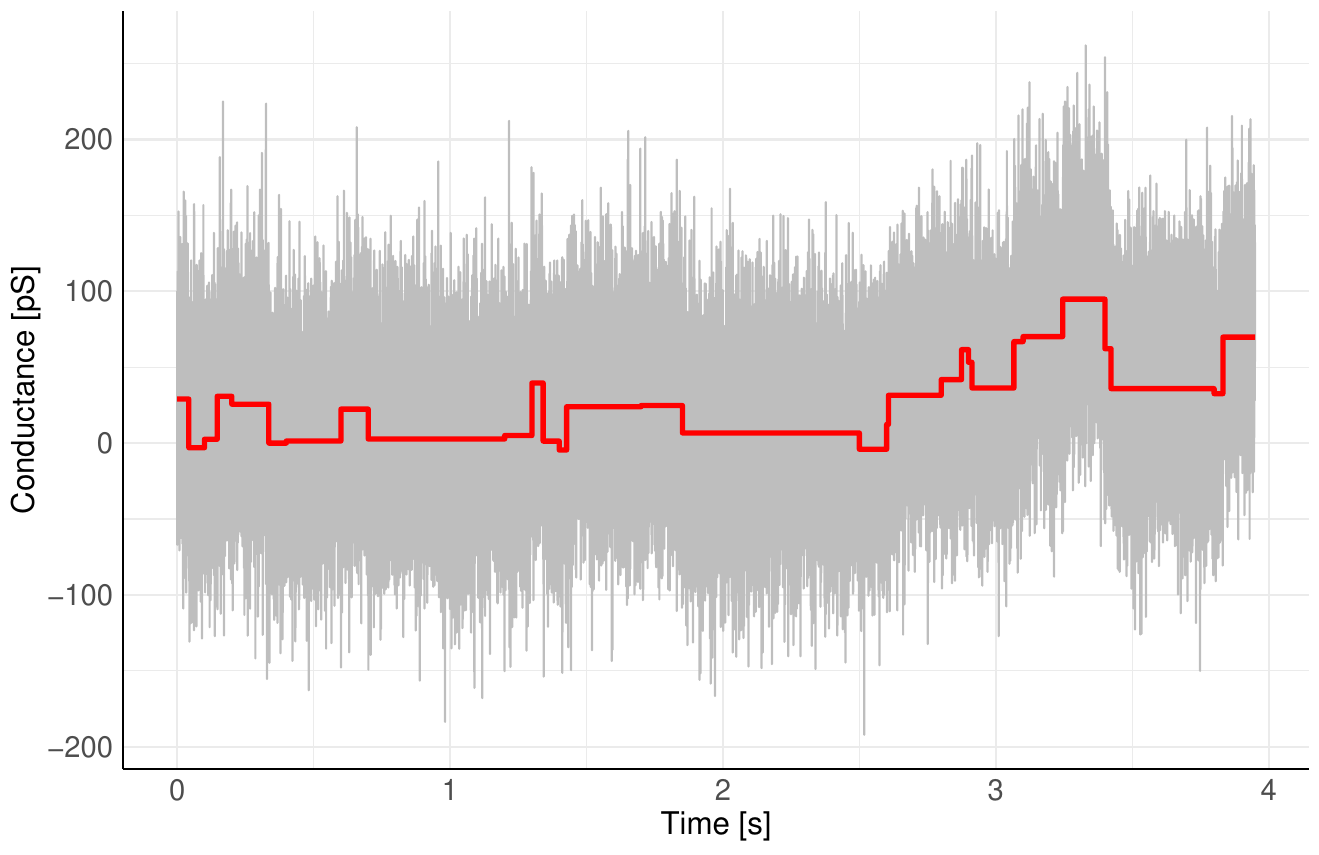  }}
\caption{Gramicidin data (grey) and idealisation (red) by MUSCLE. The estimated number of channels is 3, see~\cref{fig3}.}
\label{fig1}
\end{figure}

\begin{figure}[!t]
\centerline{\includegraphics[width=.6\textwidth]{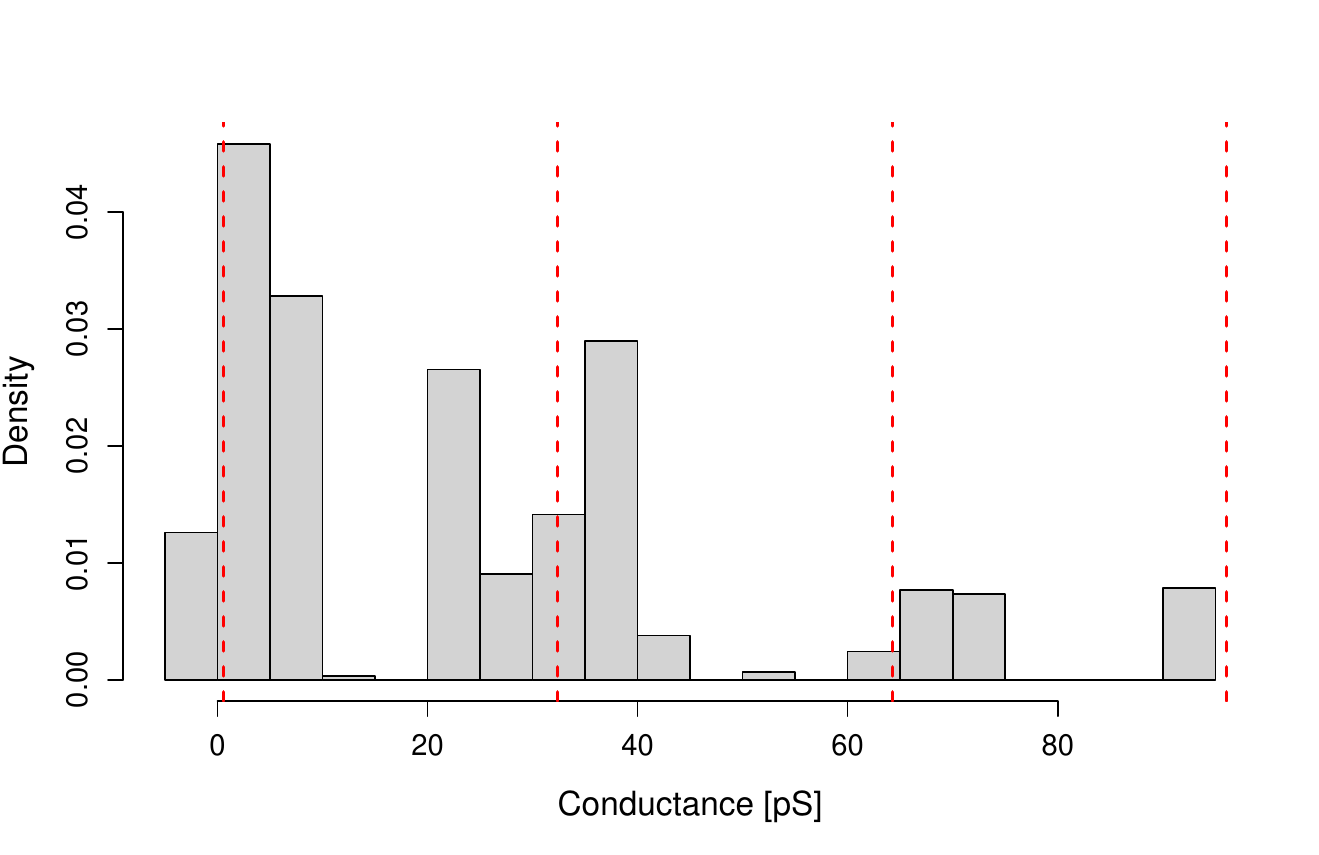}}
\caption{Histogram of the idealised conductance levels with the centres of the groups marked by vertical dashed lines.}
\label{fig3}
\end{figure}

To enhance the reliability of our finding, we repeat the above analysis with different choices for $L$. The inferred cooperative behaviour remains the same, see Appendix~\ref{analysis_different_L} for details.

\subsection{Influence of applied voltages}

Furthermore, we analyse how the voltage applied in the experiment influences the cooperative behaviour of gramicidin channels. Towards this end, we record multiple ion channels at different voltages (see \cref{s:exp}), and then analyse the channel recordings using the IDC method as in \cref{data_analysis}. The measurements and the corresponding idealisations are in \cref{fig8}. Based on the histograms of idealised conductance levels in \cref{fig9,fig10,fig11,fig12}, we set the number $L$ of channels as $2,3,5$ and $7$ for the channel recordings at the applied voltages of 50~mV, 100~mV, 150~mV and 200~mV, respectively. The minimum distance estimates of the parameter vector $\vec{\theta}$ in VND-MC are shown in \cref{table:2}.

\begin{figure}[!t]
\centerline{\includegraphics[height=10cm,width=\textwidth]{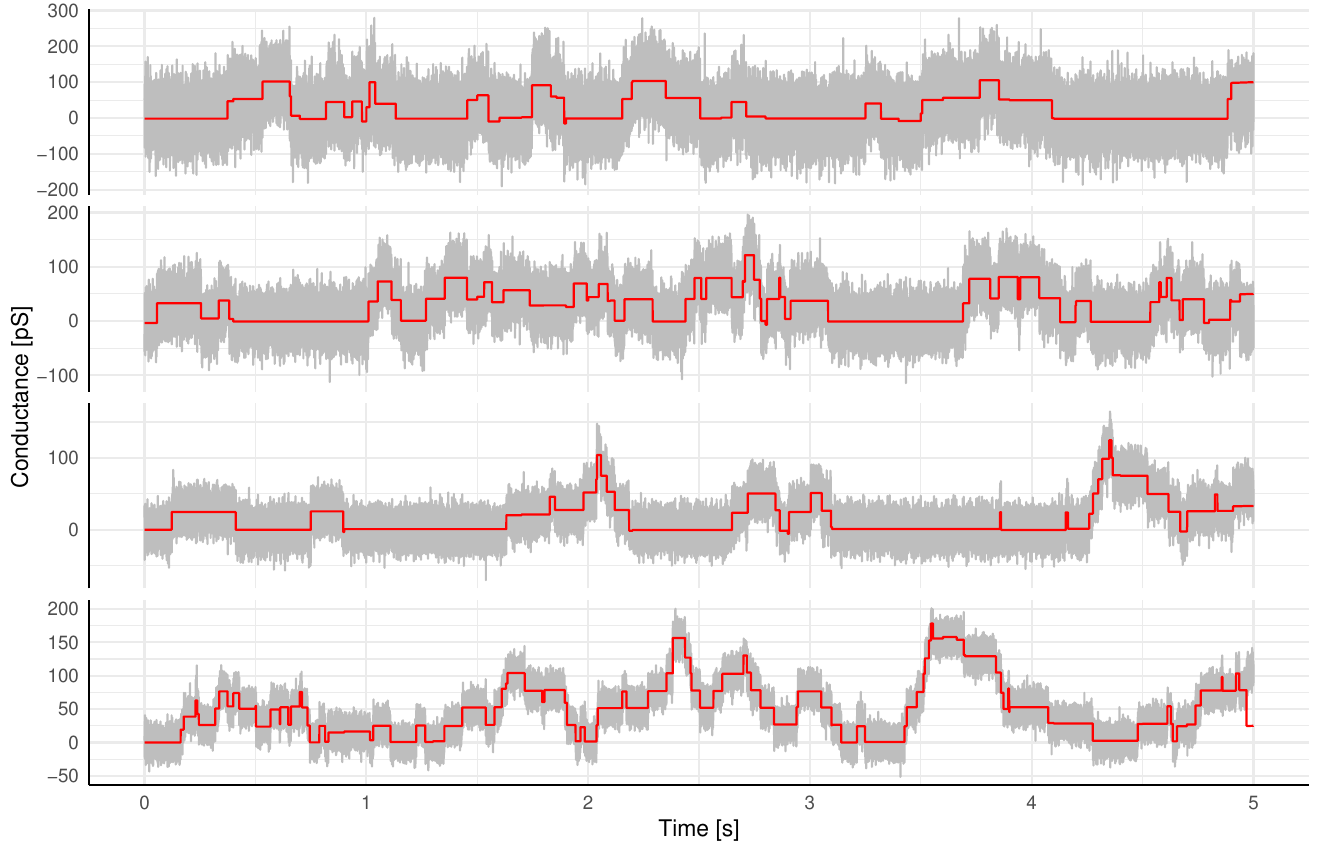}}
\caption{Gramicidin data (grey) and idealisation (red) by MUSCLE. From top to bottom, the first panel shows the scenario at the applied voltage of 50 mV, with the estimated number of channels equal to 2, see~\cref{fig9}; the second panel shows the scenario at the applied voltage of 100~mV, with the estimated number of channels equal to 3, see~\cref{fig10}; the third panel shows the scenario at the applied voltage of 150 mV, with the estimated number of channels equal to 5, see~\cref{fig11}; the last panel shows the scenario at the applied voltage of 200 mV, with the estimated number of channels equal to 7, see~\cref{fig12}.}
\label{fig8}
\end{figure}
\begin{table}[h!]
\centering
\caption{Estimation of the parameter vector in VND-MC at different applied voltages.}
\label{table:2}
\begin{tabular}{||c| c||} 
 \hline
Voltage & $\hat{\vec{\theta}}$ \\ [0.5ex] 
 \hline\hline
 50mV & (0.9999157, 0.9998241, 0.9996463, 0.9997629)  \\ 
 \hline
 100mV &  (0.9999031, 0.9998208, 0.9999376,\\
        &0.9996919, 0.9995520, 0.9995994)\\ 
 \hline
 150mV & (0.9999698, 0.9999453, 0.9999219, 0.9997571, \\&0.9990986, 0.9997111,0.9996450, 0.9998381,\\& 0.9995483, 0.9990504)\\
 \hline
 200mV & (0.9999137, 0.9999237, 0.9998968, 0.9999059,\\
    &0.9998578, 0.9998785, 0.9997587, 0.9995835,\\
    &0.9997637, 0.9997741, 0.9997079,\\ &0.9998529, 0.9999186, 0.9992131)  \\ [1ex] 
 \hline
\end{tabular}
\end{table}

\begin{figure}[!t]
\centerline{\includegraphics[width=.6\textwidth]{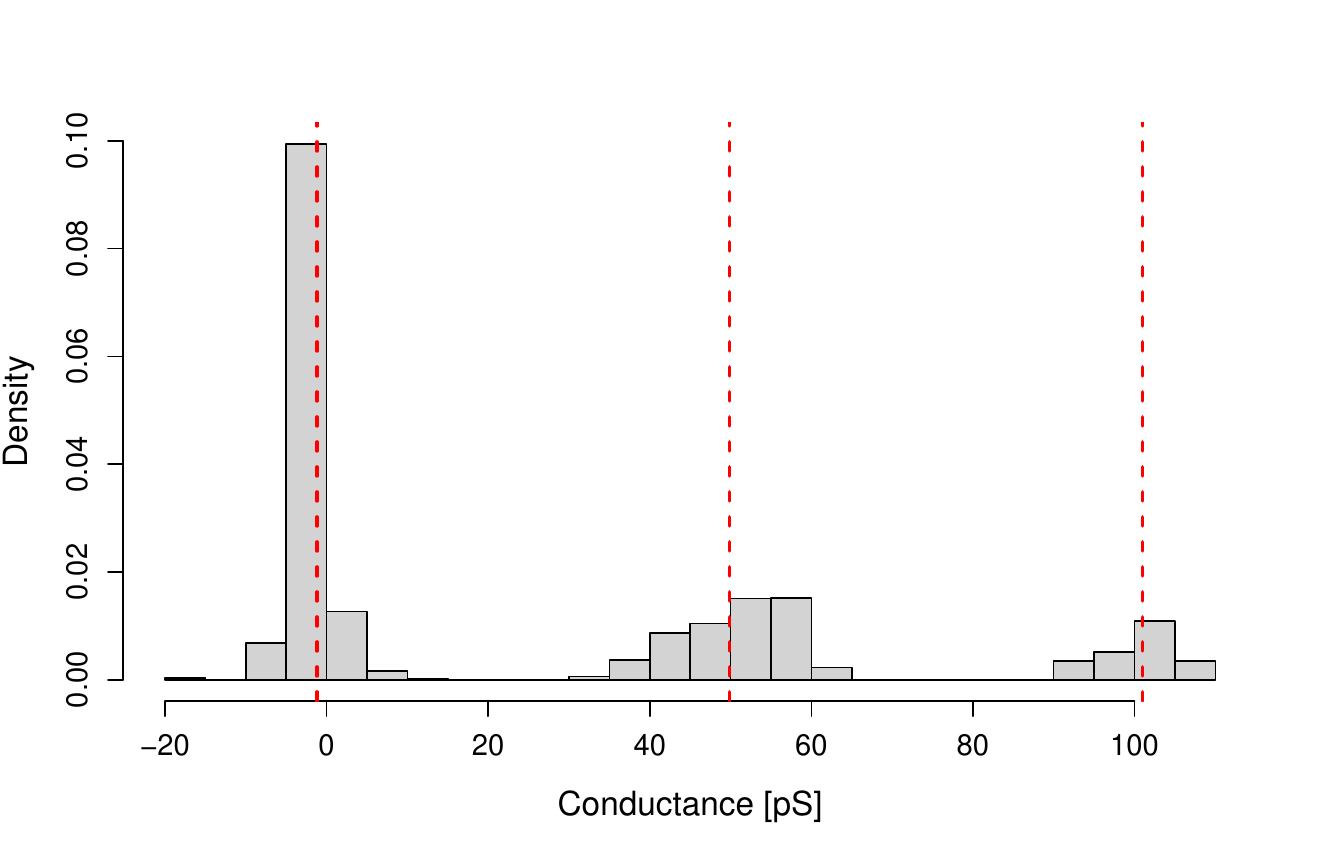}}
\caption{Histogram of the idealised conductance levels for channel recordings at 50~mV applied voltage with the centres of the groups marked by vertical dashed lines.}
\label{fig9}
\end{figure}

\begin{figure}[!t]
\centerline{\includegraphics[width=.6\textwidth]{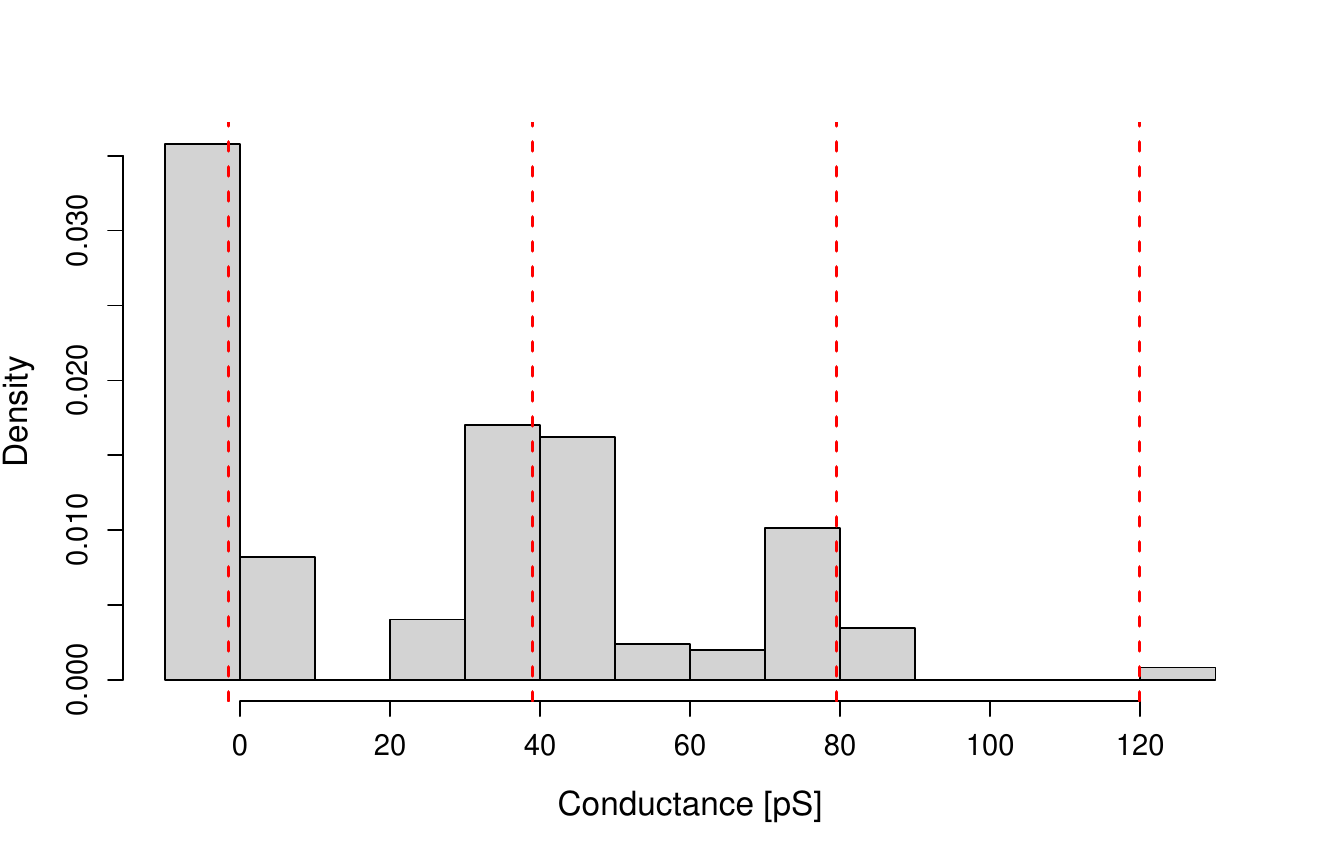}}
\caption{Histogram of the idealised conductance levels for channel recordings at 100~mV applied voltage with the centres of the groups marked by vertical dashed lines.}
\label{fig10}
\end{figure}

\begin{figure}[!t]
\centerline{\includegraphics[width=.6\textwidth]{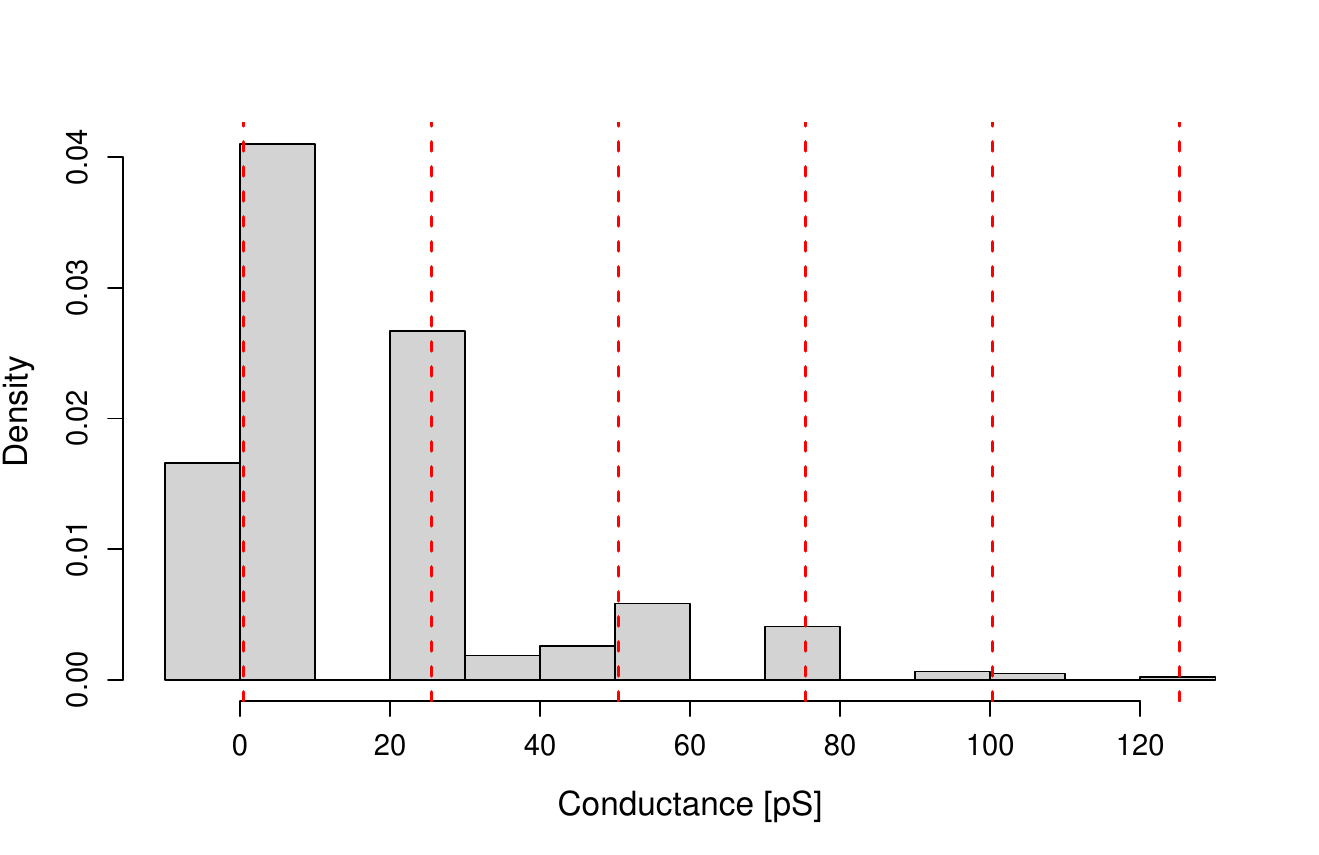}}
\caption{Histogram of the idealised conductance levels for channel recordings at 150~mV applied voltage with the centres of the groups marked by vertical dashed lines.}
\label{fig11}
\end{figure}

\begin{figure}[!t]
\centerline{\includegraphics[width=.6\textwidth]{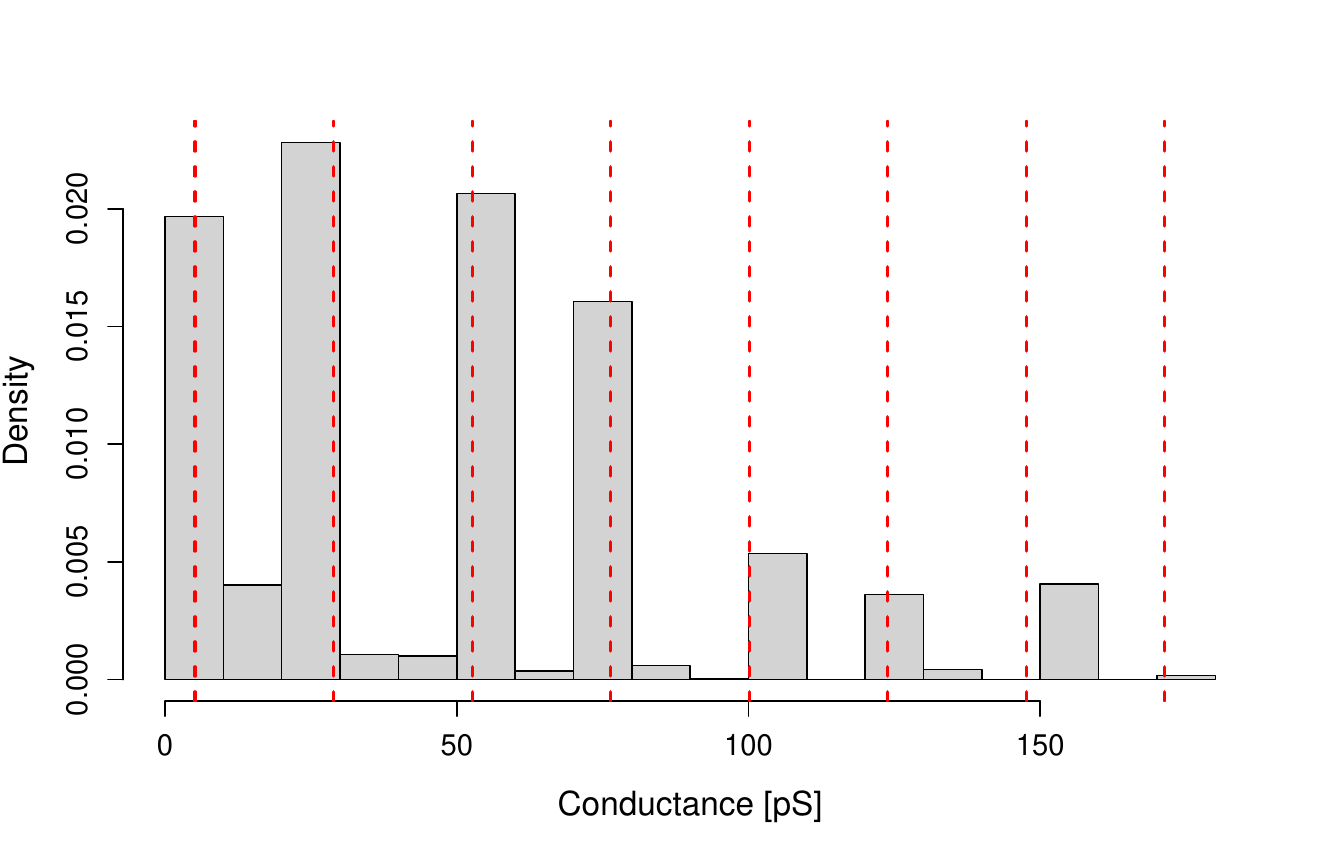}}
\caption{Histogram of the idealised conductance levels for channel recordings at 200~mV applied voltage with the centres of the groups marked by vertical dashed lines.}
\label{fig12}
\end{figure}

As the applied voltage increased, we observed certain systematic changes in the dynamics of the ion channels: the switches of states become more frequent and the number of open channels gets larger. Moreover, the signal-to-noise ratio in the channel recordings appears to increase as the applied voltage increases. Despite these changes at different applied voltages, the estimated values of the ratios in \cref{behaviour} remain very close to one. Thus, based on our data analysis, we find compelling evidence that the cooperative behaviour of gramicidin channels remains to be independent, unaffected by the varying applied voltages during the experiment.

\section{Discussion}\label{s:discuss}
We have developed a novel robust data analysis method IDC for the inference of cooperative behaviour of ion channels. It requires only  voltage-clamp current measurements of an ensemble of ion channels, which significantly reduces the experimental efforts in the study of interactions between ion channels. In comparison to existing data analysis approaches (i.e., \cite{CHUNG} and \cite{vanegas2023}), IDC incorporates the built-in low-pass filter in voltage-clamp experiments and allows a wide range of different types of noise (including outliers) as well as baseline fluctuations, which may often be encountered in ion channel recordings. This not only enhances the reliability of data analysis but also facilitates the experimental endeavour in preparing voltage-clamp current measurements. Application of IDC to other types of ion channels (e.g., KcsA) is a well-conceived direction for future research. Finally, we stress that IDC can be used to infer collective behaviour in coupled Markov models from the noisy pooled measurements, with applications to scenarios beyond multiple ion channels.

\appendices

\section{Analysis of gramicidin ion channels using different number of channels}\label{analysis_different_L}
Based on the simulation study in \cref{number_channels} we expect that our result is robust with respect to mis-specification of $L$ in the discretisation step, but to obtain further certainty, we apply IDC with different choices of $L \in\{ 2,\ldots, 6\}$. The estimated parameter vector by the minimum distance estimator is in \cref{table:1}. From this, we confirm our finding that multiple gramicidin ion channels behave independently of each other.

\begin{table}[h!]
\centering
\caption{Estimation of the parameter vector in VND-MC with different choices of $L$, the number of ion channels.}
\label{table:1}
\begin{tabular}{||c| c||} 
 \hline
 \label{table1}
$L$ & $\hat{\vec{\theta}}$ \\ [0.5ex] 
 \hline\hline
 2 & (0.9999341, 0.9999124, 0.9998522, 0.9998779)  \\ 
 \hline
 3 &  (0.9999560, 0.9999561, 0.9998173, \\&0.9998505, 0.9998387, 0.9998944)  \\
 \hline
 4 & (0.9999774, 0.9999987, 0.9999121, 0.9998176,\\&0.9997241, 0.9999660, 0.9998924, 0.9999208) \\
 \hline
 5 & (0.9999819, 0.9999995, 0.9999659, 0.9999976, \\&0.9997786, 0.9997235, 0.9999477, 0.9994520,\\&0.9999870, 0.9999690)  \\
 \hline
 6 & (0.9999734, 0.9999920, 0.9999830, 0.9999655,\\&0.9999974, 0.9997802, 0.9998935, 0.9999312,\\&0.9999654, 0.9995890, 0.9999896, 0.9999740) \\ [1ex] 
 \hline
\end{tabular}
\end{table}

\section*{Acknowledgment}
This work was supported by DFG (German Reserach Foundation) CRC~1456 {Mathematics of Experiment}, and in part by the DFG under Germany's Excellence Strategy, project EXC~2067 {Multiscale Bioimaging: from Molecular Machines to Networks of Excitable Cells} (MBExC). The authors thank Zhi Liu and Benjamin Eltzner for fruitful discussions. 


\end{document}